\pdfoutput=1
\documentclass[onecolumn,journal,draftclsnofoot,12pt]{IEEEtran}

\usepackage{bm}
\usepackage{amssymb}
\usepackage{amsmath}
\usepackage{amsfonts}
\usepackage{graphicx}
\usepackage{subfigure}
\usepackage{cite}
\usepackage{enumerate}
\usepackage{url}
\usepackage{epstopdf}
\usepackage{float}
\usepackage{color}
\usepackage{multirow}
\usepackage{enumitem}

\usepackage[normalem]{ulem}  
\newcommand{\beq}{\begin{equation}}
\newcommand{\eeq}{\end{equation}}

\allowdisplaybreaks[4]
\makeatletter  
\newif\if@restonecol  
\makeatother

\usepackage{amsmath}  

\usepackage[linesnumbered,ruled,vlined]{algorithm2e}
\usepackage{algpseudocode}

\usepackage{makecell} 

\renewcommand{\Pr}{\mathrm{Pr}}
\newcommand{\nGamma}{\bm{\mathit{\Gamma}}}

\newcommand{\rl}{\mathrm{rl}}
\newcommand{\los}{\mathrm{los}}

\newcommand{\bigO}{\mathcal{O}}
\newcommand{\proofend}{$\hfill\blacksquare$\par}

\newcommand{\sensmap}{mapping of the received signals}
\newcommand{\sensmapthe}{mapping of the received signals}

\newcommand{\config}{beamformer pattern}

\newcommand{\configmat}{control matrix}

\newcommand{\sensfunc}{mapping of the received signals}

\newcommand{\mselem}{reconfigurable element}

\newcommand{\reflect}{\mathrm{ref}}

\newcommand{\newstate}{configuration}

\newif\ifshowannote
\showannotetrue

\newcommand{\LRT}[2]{%
  \mathrel{\mathop\gtrless\limits^{#1}_{#2}}%
}

\newtheorem{proposition}{\bf Proposition}
\newtheorem{theorem}{\bf Theorem}
\newtheorem{lemma}{\bf Lemma}

\newtheorem{proof}{Proof}

\newcounter{optcnt}
\usepackage{docmute}

\usepackage[numbers,sort&compress]{natbib}
\usepackage{etoolbox}
\makeatletter
\patchcmd{\NAT@citexnum}
  {%
    \ifx\NAT@last@yr\relax
      \def@NAT@last@yr{\@citea}%
    \else
      \def@NAT@last@yr{--\NAT@penalty}%
    \fi
  }
  {%
    \def@NAT@last@yr{--\NAT@penalty}%
  }
  {}{\FAIL}

\newcommand{\revised}{\textcolor[rgb]{0,0,0}}
\usepackage{graphicx}
\usepackage{bm}
\usepackage{makecell} 
\usepackage{multirow}
\usepackage{lineno}
\usepackage{clipboard}
\usepackage{color}

\begin{document}

\title{\huge{MetaSensing: Intelligent Metasurface Assisted RF 3D Sensing by Deep Reinforcement Learning}}

\author{
\IEEEauthorblockN{
\normalsize{Jingzhi~Hu},~\IEEEmembership{\normalsize Graduate~Student~Member,~IEEE},
\normalsize{Hongliang~Zhang},~\IEEEmembership{\normalsize Member,~IEEE},
\normalsize{Kaigui~Bian},~\IEEEmembership{\normalsize Senior~Member,~IEEE},
\normalsize{Marco~Di~Renzo},~\IEEEmembership{\normalsize Fellow,~IEEE},
\normalsize{Zhu~Han},~\IEEEmembership{\normalsize Fellow,~IEEE},
and~\normalsize{Lingyang~Song},~\IEEEmembership{\normalsize Fellow,~IEEE}
}
\thanks{
 J. Hu and L. Song are with Department of Electronics, Peking University.~(email: \{jingzhi.hu, lingyang.song\}@pku.edu.cn)
 }
 \thanks{
H. Zhang is with Department of Electrical Engineering, Princeton University.~(email: hongliang.zhang92@gmail.com)
}
\thanks{K. Bian is with Department of Computer Science, Peking University.~(email: bkg@pku.edu.cn)}
\thanks{
M. Di Renzo is with Universit\'e Paris-Saclay, CNRS and CentraleSup\'elec, Laboratoire des Signaux et Syst\`emes,  91192 Gif-sur-Yvette, France.~(email: marco.direnzo@centralesupelec.fr)
}
\thanks{
Z. Han is with Electrical and Computer Engineering Department, University of Houston, and also with Department of Computer Science and Engineering, Kyung Hee University.~(email: hanzhu22@gmail.com)
}
}

\maketitle
\begin{abstract}
Using RF signals for wireless sensing has gained increasing attention.
However, due to the unwanted multi-path fading in uncontrollable radio environments, the accuracy of RF sensing is limited.
Instead of passively adapting to the environment, in this paper, we consider the scenario where an intelligent metasurface is deployed for sensing the existence and locations of 3D objects.
By programming its {\config}s, the metasurface can provide desirable propagation properties.
However, achieving a high sensing accuracy is challenging, since it requires the joint optimization of the {\config}s and mapping of the received signals to the sensed outcome.
To tackle this challenge, we formulate an optimization problem for minimizing the cross-entropy loss of the sensing outcome, and propose a deep reinforcement learning algorithm to jointly compute the optimal {\config}s and the {\sensmap}.
Simulation results verify the effectiveness of the proposed algorithm and show how the sizes of the metasurface and the target space influence the sensing accuracy.
\end{abstract}


\section{Introduction}

Recently, leveraging widespread radio-frequency~(RF) signals for wireless sensing applications has attracted growing research interest.
Different from methods based on wearable devices or surveillance cameras, RF sensing techniques need no direct contact with the sensing targets~\cite{Kianoush2017Device}.
The basic principle behind RF sensing is that the influence of the target objects on the propagation of wireless signals can be potentially recognized by the receivers~\cite{Lee2009Wireless}.
RF sensing techniques can be widely applied to many scenarios of daily life, such as surveillance~\cite{He2006Vigilnet}, crowd sensing~\cite{Ota2018QUOIN}, ambient assisted living~\cite{cook2009assessing}, and remote health monitoring~\cite{amin2016radar}.
In these applications, it is crucial to have high sensing accuracies.

Many RF-based sensing methods based on WiFi signals or millimeter wave signals have been proposed for sensing and recognizing human being and objects.
In~\cite{zhang2019Feasibility}, the authors designed an RF sensing system that can detect the location and type of moving objects by using Wi-Fi signals.
In~\cite{jiang2018towards}, the authors proposed a deep learning based RF sensing framework that can remove environmental and subject-specific information and can extract environmental/subject-independent features contained in the sensing data.
In~\cite{Hsu2019Enabling}, the authors designed a low-power RF sensing system that automatically collects behavior patterns of people.

In addition, using RF sensing to capture human beings and indoor scenes has being explored.
In~\cite{adib2015capturing,zhao2018rf}, the authors used wide-band RF transceivers with multiple-input-multiple-output~(MIMO) antennas to capture images of human skeletons and showed that it is possible to reconstruct the human skeleton even when the RF signals are blocked by walls.
In~\cite{PedrossEngel2018Orthogonal}, the authors proposed to use mutually orthogonally coded millimeter wave signals to image the scenes including human beings and objects.
However, using RF signals for sensing usually encompasses a signal collection and analysis process which passively accept the radio channel environment.
The radio environment is unpredictable and usually unfavorable, and thus the sensing accuracy of conventional RF sensing methods is usually affected by unwanted multi-path fading~\cite{honma2018Human, li2020programmable}, and/or unfavorable propagation channelsd from the RF transmitters to the receivers.

Intelligent metasurfaces have been proposed as a promising solution for turning unwanted propagation channels into favorable ones~\cite{Renzo2019Smart, Gacanin2020Wireless}.
A metasurface is composed of a large number of electrically reconfigurable elements, which applies different phase-shifts on the RF signals that impinge upon it~\cite{Basar2019Generalization, Zhang2020Reflective}.
By programming the reconfigurable elements, a metasurface deployed in the environment can change the RF propagation channel and create favorable signal beams for sensing\cite{ElMossallamy2020Reconfigurable}.
We refer to the codings of the {\mselem}s as the \emph{{\config}s}.
Through dynamically designing the {\config}s, a metasurface can actively control the RF signal beams in the sensing process, which potentially improves the sensing accuracy.
Instead of employing complex and sophisticated RF transmitters and receivers~\cite{Hashida2020Intelligent}, metasurface assisted RF sensing paves a new way of developing RF sensing methods, which have the capabilities of controlling, programming, and hence customizing the wireless channe.
In literature, the authors of~\cite{Li2019Machine} explored the use of metasurfaces to assist RF sensing and obtain $2$D images for human beings.
Besides, in~\cite{Zhang2020Towards}, the authors proposed a metasurface assisted RF system to obtain localization of mobile users.
Nevertheless, no research works have tackled the analysis and design of metasurface assisted $3$D RF sensing, which is more challenging to analyze and optimize than $2$D RF sensing.

In this paper, we consider a metasurface assisted RF $3$D sensing scenario, which can sense the existence and locations of $3$D objects in a target space.
Specifically, by programming the {\config}s, the metasurface performs beamforming and provides desirable RF propagation properties for sensing.
However, there are two major challenges in obtaining high sensing accuracy in metasurface assisted RF sensing scenarios.
\begin{itemize}[leftmargin=*]
\item \emph{First}, the {\config}s of the metasurface need to be carefully designed to create favorable propagation channels for sensing. 
\item \emph{Second}, the \emph{\sensmapthe}, i.e., the mapping from the signals received at the RF receiver to the sensing results of the existence and locations of the objects, needs to be optimized as well.
\end{itemize}
Nevertheless, the complexity of finding the optimal {\config}s is extremely high because the associate optimization problem is a discrete nonlinear programming with a large number optimization variables.
Besides, the optimization of the {\config}s and the {\sensmap} are closely coupled together, which makes optimizing the sensing accuracy in metasurface assisted RF sensing scenarios even harder.

To tackle these challenges, we formulate an optimization problem for  sensing accuracy maximization by minimizing the cross-entropy loss of the sensing results with respect to the {\config}s and the {\sensmap}.
In order to solve the problem efficiently, we formulate a Markov decision process~(MDP) for the optimization problem and propose a deep reinforcement learning algorithm.
The proposed deep reinforcement learning algorithm is based on the policy gradient algorithm~\cite{sutton1998reinforcement} and is referred to as the progressing reward policy gradient~(PRPG) algorithm, since the reward function of the MDP is consistently being improved during the learning process.
The computational complexity and the convergence of the proposed algorithm are analyzed.
Moreover, we derive a non-trivial lower-bound for the sensing accuracy for a given set of {\config}s of the metasurface.
Simulation results verify the effectiveness of the proposed algorithm and showcase interesting performance trends about the sensing accuracy with respect to the sizes of the metasurface and the target space.
In particular, the contributions of this paper can be summarized as follows.
\begin{itemize}[leftmargin=*]
    \item  We consider a metasurface assisted RF sensing scenario which can sense the existence and locations of objects in a $3$D space. 
		Then, we formulate an optimization problem to minimize the cross-entropy loss of the sensing results through optimizing the {\config}s and the {\sensmap}.  To this end, we adopt a MDP-based framework.
    \item We propose a deep reinforcement learning algorithm named PRPG to solve the formulated MDP.
		The complexity and the convergence of the proposed algorithm are analyzed, and a non-trivial lower-bound for the sensing accuracy is derived.
	\item We use simulation results to verify that the proposed algorithm outperforms other benchmark algorithms in terms of training speed and sensing accuracy. 
		The simulation results unveil trends about the sensing accuracy as a function of the sizes of the metasurface and the target space, which gives insights on the implementation of practical metasurface assisted RF sensing systems.
\end{itemize}

The rest of this paper is organized as follows. 
In Section~\ref{sec: overview}, we introduce the model of the metasurface assisted RF sensing scenario. 
In Section~\ref{sec: problem formulation}, we formulate the optimization problem to optimize the sensing accuracy by minimizing the cross-entropy loss of the sensing results.
In Section~\ref{sec: algorithm design}, we formulate an MDP for the optimization problem and then proposed the PRPG algorithm to solve it.
In Section~\ref{algorithm analysis}, the complexity and convergence of the PRPG algorithm are analyzed, and a lower-bound for the sensing accuracy is derived.
Simulation results are provided in Section~\ref{sec: simulation result} and conclusions are drawn in Section~\ref{sec: conclu}.

\section{System Model}
\label{sec: overview}
In this section, we introduce the metasurface assisted $3$D RF sensing scenario, which is illustrated in Fig.~\ref{fig: sys mod}.
In this scenario, there exist a pair of single-antenna RF transceivers, a metasurface, and a target space where the objects are located.
The metasurface reflects and modifies the incident narrow-band signals at a certain frequency $f_c$.
The Tx unit and Rx unit of the transceiver keep transmitting and receiving at $f_c$.
The target space is a cubical region that is discretized into $M$ equally-sized \emph{space grids}.
Each space grid is of size $\Delta l_x \times \Delta l_y \times \Delta l_z$.

The sensing process adopted in the considered scenario can be briefly described as follow.
The signals transmitted by the Tx unit are reflected and modified by the metasurface before entering into the target space.
The modified signals are further reflected by the objects in the target space and received by the Rx unit.
Then, the Rx unit maps the received signals to the sensing result, which indicates whether an object exists in each space grid.

In the following, we introduce the metasurface model in Subsection~A, the channel model accounting for the metasurface in Subsection~B, and the sensing protocol in Subsection~C.

\begin{figure}[!t]
\center{\includegraphics[width=0.5\linewidth]{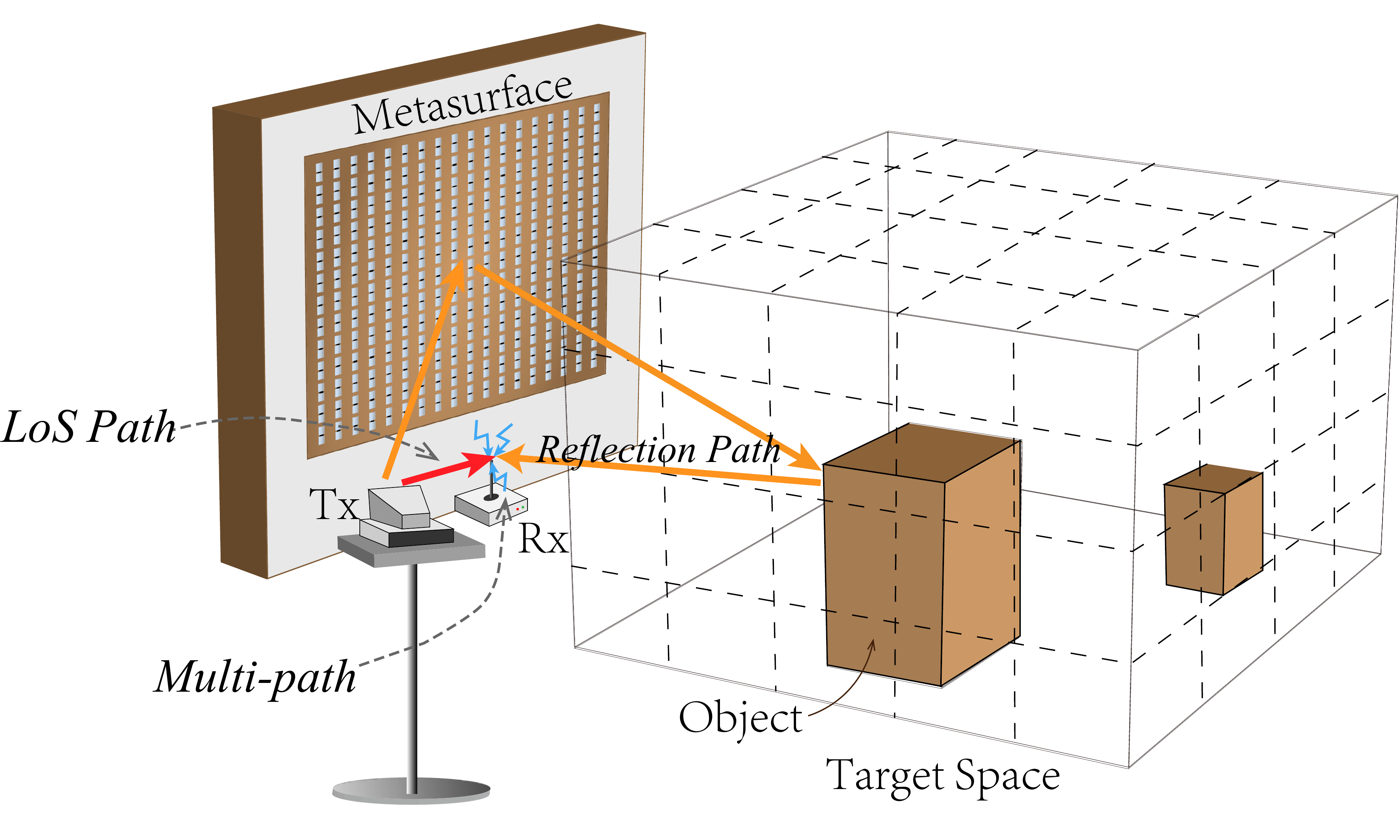}}
	\vspace{-0.5em}
	\setlength{\belowcaptionskip}{-1.em} 
	\caption{Illustration of the metasurface assisted RF sensing scenario.}
	\vspace{-1.em}
	\label{fig: sys mod}
\end{figure}

\subsection{Metasurface Model}
\label{ssec: meta model}
A metasurface is an artificial thin film of electromagnetic reconfigurable materials, which is composed of uniformly distributed {\mselem}s~\cite{Liu2019Reconfigurable}.
As shown in Fig.~\ref{fig: sys mod}, the {\mselem} of the metasurface are arranged in a two-dimensional array. 
By controlling the positive-intrinsic-negative~(PIN) diodes coupled with each {\mselem}, the {\mselem} can adjust its electromagnetic response to the incident RF signals.
\revised{For each {\mselem}, we refer to the different responses to incident RF signals as the {\mselem}'s \emph{\newstate} as in~\cite{Dai2020Reconfigurable}.
By changing the {\newstate} of each {\mselem}, the metasurface is able to modify the reflected signals and perform beamforming~\cite{di2019hybrid}.}

We assume that each {\mselem} has $N_S$ {\newstate}s, and each {\newstate} of an element has a unique reflection coefficient for the incident RF signals.
To be specific, we assume that each row and column of the metasurface contain the same number of {\mselem}s, and the total number of {\mselem}s is denoted by $N$. 
Based on~\cite{our_ris_work}, we denote the reflection coefficient of the $n$-th {\mselem} corresponding to the incident signal from the TX unit and the reflected signal towards the $m$-th space grid by $r_{n,m}(c_n)$.
Here, $c_n\in[1,N_S]$ denotes the {\newstate} of the $n$-th {\mselem} and $c_n\in \mathbb Z$, where $\mathbb Z$ denotes the set of integers.

\subsection{Channel Model}
\label{ssec: channel model}
In the metasurface assisted RF sensing scenario, the Tx unit and Rx unit adopt single antennas to transmit and receive RF signals.
The Tx antenna is a directional antenna, which points towards the metasurface so that most of the transmitted signals are reflected by the metasurface and propagate into the target space.
The signals reflected by the metasurface are reflected by the objects in the target space and then reach the Rx antenna.
The Rx antenna is assumed to be omni-directional and located right below the metasurface, as shown in Fig.~\ref{fig: sys mod}.
This setting ensures that the signals reflected by the metasurface are not  directly received by the Rx antenna, and thus most of the received signals contain the information of the objects in the target space.

As shown in Fig.~\ref{fig: sys mod}, the transmission channel from the Tx antenna to the Rx antenna is composed of three types of paths, i.e., the line-of-sight~(LoS) path, the reflection paths, and the environmental scattering paths.
The LoS path indicates the direct signal path from the Tx antenna to the Rx antenna.
The reflection paths are the paths from the Tx antenna to the Rx antenna via the reflections from the metasurface and the objects in the target space.
The environmental scattering paths account for the signals paths between the Tx antenna and the Rx antenna which involve complex reflection and scattering in the surrounding environment.
Then, the equivalent baseband representation of the received signal containing the signals from all these three types of paths is denoted by $y$ and can be expressed as
\beq
\label{equ: received signal}
y= h_{\mathrm{los}}\cdot \sqrt{P}\cdot x + \sum_{m=1}^M\sum_{n=1}^N h_{n,m}(c_n,\nu_m) \cdot \sqrt{P} \cdot x + h_{\rl}\cdot \sqrt{P} \cdot x +\sigma,
\eeq
where $P$ is the transmit power, and $x$ denotes the transmitted symbol. 

The component terms of (\ref{equ: received signal}) can be explained in detail as follows.
The first term, i.e., $h_{\mathrm{los}}\cdot P\cdot x$, corresponds to the signal received in the LoS path, where $h_\los$ denotes the gain.
Based on~\cite{goldsmith2005wireless}, $h_\los$ can be expressed as
\beq
h_\los = \frac{\lambda}{4\pi} \cdot \frac{\sqrt{g_{T}g_{R}}\cdot e^{-j2\pi d_{\mathrm{los}}/\lambda}}{d_{\mathrm{los}}},
\eeq
where $\lambda$ is the wavelength of the signal, 
$g_{T}$ and $g_{R}$ denote the gains of the Tx and Rx antennas, respectively,
and $d_{\mathrm{los}}$ is the distance from the Tx antenna to the Rx antenna.

The second term in~(\ref{equ: received signal}) corresponds to the signals that reach the Rx antenna via $N\cdot M$ reflection paths.
\revised{In the second term, $h_{n,m}(c_n, \nu_m)$ denotes the gain of the reflection path via the $n$-th {\mselem} in {\newstate} $c_n$ and the $m$-th space grid with reflection coefficient $\nu_m$.}
Based on~\cite{tang2019wireless,di2019hybrid}, $h_{n, m}(c_n, \nu_m)$ can be formulated as follows
\beq
\label{equ: main channel gain}
h_{n, m}(c_n, \nu_m) \!=\! \frac{\lambda^2 \!\cdot\! r_{n,m}(c_n)\!\cdot\! \nu_{m}\!\cdot\!  \sqrt{g_{T}g_{R}}\!\cdot\! e^{-j2\pi (d_n+d_{n,m})/\lambda}}{(4\pi)^2\cdot d_{n} \cdot d_{n,m}},
\eeq
where $d_n$ denotes the distance from the Tx antenna to the $n$-th {\mselem} and $d_{n,m}$ denotes the distance from the $n$-th {\mselem} to the Rx antenna via the center of the $m$-th space grid.

Finally, the third and forth terms in~(\ref{equ: received signal}) correspond to the signals from the environmental scattering paths and the additive noise at the Rx antenna, respectively.
The symbol $h_{\rl}\in\mathbb C$ denotes the equivalent gain of all the environmental scattering paths, and $\sigma$ is a random signal that follows the complex normal distribution, $\sigma\sim\mathcal{CN}(0,\epsilon)$ with $\epsilon$ being the power of the noise.

Moreover, we refer to the vector of {\newstate}s selected for the $N$ {\mselem}s as a \emph{\config} of the metasurface, which can be represented by a $N\times N_S$-dimensional binary row vector $\bm c = (\hat{\bm o}(c_1), ..., \hat{\bm o}(c_N))$.
Specifically, $\hat{\bm o}(i)$~($\forall i\in[1,N_S]$) denotes the $N_S$-dimensional row vector whose $i$-th element is $1$ and the other elements are $0$.
Based on the definition of the {\config}, the received signal in~(\ref{equ: received signal}) can be reformulated as
\beq
\label{equ: received signal matrix form}
y= h_{\mathrm{los}}\cdot \sqrt{P}\cdot x + \bm c\bm A\bm \nu \cdot \sqrt{P} \cdot x + h_{\rl}\cdot \sqrt{P} \cdot x +\sigma,
\eeq
where $\bm \nu=(\nu_1, \dots, \nu_M)$ denotes the vector of reflection coefficients of the $M$ space grids, $\bm A = (\bm \alpha_1,\dots, \bm \alpha_M)$ is referred to as the \emph{projection matrix}, and $\bm \alpha_m = (\hat{\bm \alpha}_{m,1}, \dots, \hat{\bm \alpha}_{m,N})^T$ with $\hat{\bm \alpha}_{m,n} = (\hat{\alpha}_{m,n,1},\dots,\hat{\alpha}_{m,n,N_S})$.
Here, for all $m\in[1,M]$, $n\in[1,N]$, and $i\in[1,N_S]$, $\hat{\alpha}_{m,n,i}$ denotes the channel gain of the reflection path via the $n$-th {\mselem} in {\newstate} $i$ and the $m$-the space grid with a unit reflection coefficient, which can be expressed as follows based on~(\ref{equ: main channel gain}).
\begin{equation}
\label{equ: added alpha def}
\hat{\alpha}_{m,n,i} = \frac{\lambda^2 \cdot r_{n,m}(i)\cdot \sqrt{g_T g_R}}{(4\pi)^2 d_n d_{n,m} } \cdot e^{-j2\pi (d_n+d_{n,m})/\lambda }.
\end{equation}

\subsection{RF Sensing Protocol}
\label{ssec: rf sensing protocol}
\begin{figure}[!t] 
\center{\includegraphics[width=0.7\linewidth]{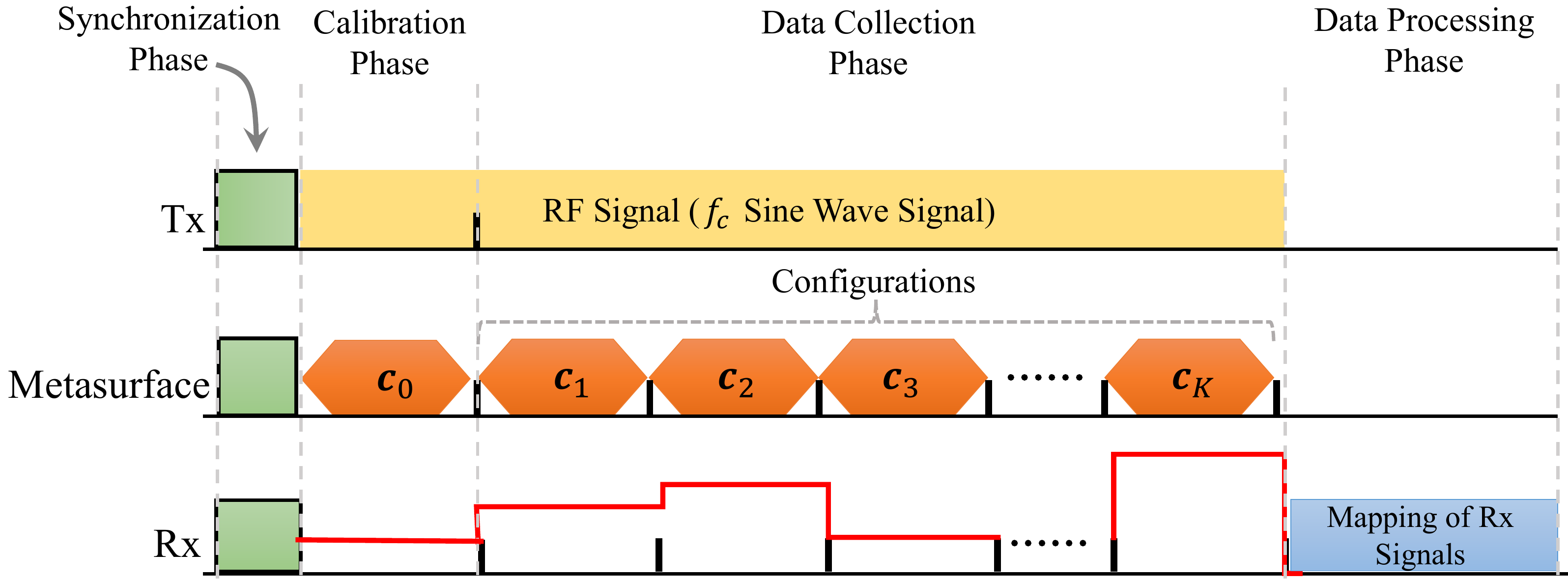}}
	\vspace{-0.5em}
	\setlength{\belowcaptionskip}{-1.em}   
	\caption{A cycle of the RF sensing protocol.}
	\label{fig: protocol}
	\vspace{-1.em}
\end{figure}

To describe the RF sensing process in the metasurface assisted scenario clearly, we formulate the following \emph{RF sensing protocol}.
In the protocol, the timeline is slotted and divided into \emph{cycles}, and the Tx unit, the Rx unit, and the metasurface operate in a synchronized and periodic manner.
As shown in Fig.~\ref{fig: protocol}, each cycle consists of four phases: a \emph{synchronization phase}, a \emph{calibration phase}, a \emph{data collection phase}, and a \emph{data processing phase}.
During the synchronization phase, the Tx unit transmits a synchronization signal to the metasurface and to the Rx unit, which identifies the start time of a cycle.

Then, in the calibration phase, the Tx unit transmits a narrow band constant signal, i.e., symbol $x$, at frequency $f_c$. 
The metasurface sets the {\config} to be $\bm c_0 = (\hat{\bm o}(1),\ldots,\hat{\bm o}(1))$, i.e., the $N$ {\mselem}s are in their first/default {\newstate}.
Besides, the received signal of the Rx unit is recorded as $y_0$.

The data collection phase is divided into $K$ \emph{frames} that are evenly spaced in time.
During this phase, the Tx unit continuously transmits the narrow band RF signal, while the metasurface changes its {\config} at the end of each frame.
As shown in Fig.~\ref{fig: protocol}, we denote the {\config}s of the metasurface corresponding to the $K$ frames by binary row vectors $\bm c_1,\ldots\bm c_K$.
Specifically, the $K$ {\config}s of the metasurface during the data collection phase constitutes the \emph{\configmat}, which is denoted by $\bm C = (\bm c_1^T,..., \bm c_K^T)^T$.
Besides, as $\bm c_k$  is a binary row vector, {\configmat} is a binary matrix.

To remove the signal form the LoS path which contains no information of the target space,  the received signals in the $K$ frames are subtracted by $y_0$.
The $K$ differences constitute the \emph{measurement vector},  which is a noisy linear transformation of $\bm \nu$ by the matrix $\nGamma$, i.e.,
\beq
\label{equ: measurement vector}
\tilde{\bm y} = \bm y - \bm y_0 = \nGamma \bm \nu + \tilde{\bm \sigma},
\eeq
where $\nGamma =\sqrt{P}\cdot x \cdot (\bm C - \bm C_0)\bm A $ with $\bm C_0 = (\bm c_0^T, \ldots, \bm c_0^T)^T$,
$\bm y$ is a $K$-dimensional vector consisting of the sampled received signals during the $K$ frames that can be calculated by~(\ref{equ: received signal matrix form}),
$\bm y_0$ is a $K$-dimensional vector with all the elements being $y_0$,
and $\tilde{\bm \sigma}$ is the difference between the noise signals and environmental scattering signals of $\bm y$ and $\bm y_0$.
\revised{In this article, we assume that the environment in the considered scenario is static or changing slowly. 
In this case, the signals from the environmental scattering paths, i.e., $h_{\mathrm{rl}}\cdot\sqrt{P}\cdot x$ is subtracted in~(\ref{equ: measurement vector}), and $\tilde{\bm \sigma}$ contains the difference between the Gaussian noise signals of $\bm y$ and $\bm y_0$.\footnote{
\revised{If the environment is changing rapidly, $h_{\mathrm{rl}}\cdot\sqrt{P}\cdot x$ can be considered as an additional complex Gaussian noise~\cite{zhang2019reconfigurable}, and $\tilde{\bm \sigma}$ in (\ref{equ: measurement vector}) is composed of the difference of the noise signals at the Rx and that of the environmental scattering signals, and thus its variance is $2\epsilon + 2\epsilon_{\mathrm{hl}}$.}
}
Specifically, the $k$-th element of $\tilde{\bm \sigma}$ is $\tilde{\sigma}_k\sim\mathcal{CN}(0, 2\epsilon)$.}
We refer to $\tilde{\bm y}$ as the \emph{measurement vector}.
Since $\nGamma$ determines how the reflection characteristics of the objects are mapped to the measurement vector, we refer to $\nGamma$ as the \emph{measurement matrix}.

Finally, during the data processing phase, the receiver maps the measurement vector obtained in the data collection phase to the sensing results, which is a vector indicating the probabilities that objects exist in the $M$ space grids.
Given {\configmat} $\bm C$, the mapping is modeled through a parameterized function, i.e., $\hat{\bm p} = \bm f^{\bm w}(\tilde{\bm y})$ with $\bm w$ being the parameter vector that is referred to as the \emph{\sensfunc}.
\revised{Moreover, the result of the mapping, i.e., $ \hat{\bm p}$, is an $M$-dimensional real-valued vector.
Specifically, its $m$-th element, i.e., $\hat{p}_{m}\in[0,1]$, indicates the probability that an object exists at the $m$-th space grid; therefore $(1-\hat{p}_{m})$ indicates the probability that the $m$-th space grid is empty.}

\section{Problem Formulation}
\label{sec: problem formulation}
In this section, we formulate the optimization problem for maximizing the sensing accuracy for the considered scenario.
We adopt the \emph{cross-entropy loss} as the objective function to measure the sensing accuracy, as minimizing the cross-entropy loss function can significantly improve the accuracy of classification and prediction~\cite{goodfellow2016deep}.
In other words, the sensing accuracy is inversely proportional to the cross-entropy loss.

\revised{
We define the cross-entropy loss in the considered scenario as 
\begin{align}
\label{equ: def of sensing loss}
L_{\mathrm{CE}} = 
    - \mathbb E_{ \bm \nu \in \mathcal V } \Big[
        & \sum_{m=1}^M  
            p_{m}(\bm\nu) \cdot \ln(\hat{p}_{m}) + (1-p_{m}(\bm\nu)) \cdot \ln(1-\hat{p}_{m}) 
    \Big], 
\end{align}
where $\mathcal V$ denotes the set of all possible reflection coefficient vectors corresponding to the existence of objects in the target space,
and $p_m(\bm \nu)$ is a binary variable indicating the object existence in the $m$-th space grid. 
Specifically, $p_m(\bm \nu)$ can be expressed as
\beq
\label{equ: prob vector for space grid}
p_{m}(\bm \nu) =
\begin{cases}
 	0,\quad \text{if $|\nu_m| = 0$},\\
 	1, \quad \text{otherwise}.
 \end{cases}
\eeq}

\revised{
In~(\ref{equ: def of sensing loss}), $\hat{\bm p}$ is determined by $\bm f^{\bm w}(\tilde{\bm y})$.
Generally, parameterized function $\bm f^{\bm w}(\tilde{\bm y})$ can take any form. 
For example, it can be a linear function, i.e., $\bm f^{\bm w}(\tilde{\bm y}) = \bm W \tilde{\bm y} + \bm w'$, where $\bm W$ and $\bm w'$ are determined by $\bm w$ and obtained by minimizing the mean squared error of the sensing results~\cite{Boyd_CONVEX}.
Besides, $\bm f^{\bm w}(\tilde{\bm y})$ can also be a nonlinear decision function, which determines the sensing results of $\tilde{\bm y}$ by using conditional probabilities~\cite{McDonough_SIGNAL}.
In this paper, we consider that $\bm f^{\bm w}(\tilde{\bm y})$ is nonlinear and modeled as a neural network, where the elements of $\bm w$ stand for the weights of the connections and the biases of the nodes.
We refer to the neural network for $\bm f^{\bm w}(\tilde{\bm y})$ as the \emph{sensing network}.
}

The optimization problem for the metasurface assisted scenario that maximizes the sensing accuracy can be formulated as the following cross-entropy minimization problem, where the {\configmat} and the {\sensmap} parameter are the optimization variables, i.e.,
\begin{align}
\text{(P1)}: \min_{\bm C, \bm w}~
& L_{\mathrm{CE}}(\bm C, \bm w),   \label{equ-obj func general opt} \\
s.t.~
&(\hat{p}_1,...,\hat{p}_M) =  \bm f^{\bm w}(\tilde{\bm y}),
\label{equ-1 const general opt}\\
&\tilde{\bm y} = \sqrt{P}\cdot x \cdot (\bm C - \bm C_0)\bm A + \tilde{\bm \sigma},
\label{equ-1.5 const general opt}\\
& \bm C = (\bm c_1^T,..., \bm c_K^T)^T,
\label{equ-2 const general opt}\\
& \bm c_{k} = (\hat{\bm o}(c_{k,1}), ..., \hat{\bm o}(c_{k,N})), ~\forall k\!\in\![1,K],
\label{equ-3 const general opt}\\
& c_{k,n} \in [1,N_S],~\forall k\!\in\![1,K], n\!\in\![1,N].
\label{equ-4 const general opt}
\end{align}
\revised{In (P1), (\ref{equ-obj func general opt}) indicates that the objective is to minimize the cross-entropy loss by optimizing $\bm C$ and $\bm w$. 
As $\hat{\bm p}$ is determined by $\bm f^{\bm w}(\tilde{\bm y})$ and $\tilde{\bm y}$ is determined by control matrix $\bm C$, $L_{\mathrm{CE}}$ defined in~(\ref{equ: def of sensing loss}) can be expressed as a function of $\bm C$ and $\bm w$.
Constraint~(\ref{equ-1 const general opt}) indicates that the probabilities for the $M$ space grids to contain objects are calculated by the {\sensfunc}, i.e., $\bm f^{\bm w}(\tilde{\bm y})$.
Constraint~(\ref{equ-1.5 const general opt}) indicates that the measurement vector is determined by {\configmat} $\bm C$ as in~(\ref{equ: measurement vector}).
Besides, constraints (\ref{equ-2 const general opt})$\sim$(\ref{equ-4 const general opt}) are due to the definition of the {\configmat} in Section~\ref{ssec: rf sensing protocol}.}
\revised{Since the {\configmat} is a binary matrix and $\bm w$ is a real-valued vector, (P1) is a mixed-integer optimization problem and is NP-hard.}

To tackle it efficiently, we decompose (P1) into two sub-problems, i.e., (P2), and (P3), as follows:
\begin{align}
\label{prob: p2}
&\text{(P2)}:
\min_{\bm w}~ L_{\mathrm{CE}}(\bm C, \bm w) ,\quad s.t.~\text{(\ref{equ-1 const general opt})}.\\
&\text{(P3)}: \min_{\bm C}~L_{\mathrm{CE}}(\bm C, \bm w), \quad s.t.~\text{(\ref{equ-1.5 const general opt}) to~(\ref{equ-4 const general opt})}.
\label{prob 3-3 const}
\end{align}
In (P2), we minimize the cross-entropy loss by optimizing $\bm w$ given $\bm C$, and in (P3), we minimize the cross-entropy loss by optimizing $\bm C$ given $\bm w$.
Based on the alternating optimization technique~\cite{bezdek2003convergence}, a locally optimal solution of~(P1) can be solved by iteratively solving (P2) and (P3).
Nevertheless, given $\bm w$, (P3) is still hard to solve due to the large number of integer variables in the {\configmat}.
Moreover, the number of iterations for solving~(P2) and~(P3) can be large before converging to the local optimum of~(P1).
If traditional methods, such as exhaustive search and branch-and-bound algorithms, are applied, they will result in a high computational complexity.
To solve~(P2) and~(P3) efficiently, we develop an MDP framework and solve it by proposing an PRPG algorithm, which are discussed in the next section. 
Furthermore, the convergence of the proposed algorithm to solve~(P1) is analyzed in Section~\ref{algorithm analysis}.

\section{Algorithm Design}
\label{sec: algorithm design}

In this section, we formulate an MDP framework for (P2) and (P3) in Subsection~A and propose a deep reinforcement learning algorithm named PRPG to solve it in Subsection~B. 

\subsection{MDP Formulation}
\label{ssec: mdp formulation}

\revised{In (P3), the optimization variable $\bm C$ is composed of a large number of binary variables satisfying constraints (\ref{equ-2 const general opt})$\sim$(\ref{equ-4 const general opt}), which makes (P3) an integer optimization problem which is NP-hard and difficult to solve.
Nevertheless, the metasurface can be considered as an intelligent agent who determines the {\newstate} of each {\mselem} for each {\config} sequentially, and is rewarded by the negative cross-entropy loss.
In this regard, the integer optimization problem (P3) can be considered as a decision optimization problem for the metasurface, which can be solved efficiently by the deep reinforcement learning technique, since it is efficient to solve highly-complexed decision optimization problems for intelligent agents~\cite{Volodymyr2015Human,Li2019Deep}.
As the deep reinforcement learning algorithm requires the target problem to be formulated as an MDP, we formulate~(P2) and~(P3) as an MDP, so that we can solve them by proposing an efficient deep learning algorithm.}

An MDP encompasses an environment and an agent, and consists of four components: the set of states $\mathcal S$, the set of available actions $\mathcal A$, the state transition function $\mathcal T$, and the reward function $\mathcal R$~\cite{sutton1998reinforcement}.
The states in $\mathcal S$ obey the Markov property, i.e., each state only depends on the previous state and the adopted action. 
Suppose the agent takes action $a$ in state $s$, and the consequent state $s'$ is given by the transition function $\mathcal T$, i.e., $\mathcal \bm s' = T(\bm s,a)$.
After the state transition, the agent receives a reward that is determined by reward function $\mathcal R$, i.e., $\mathcal R(\bm s',\bm s,a)$.

To formulate the MDP framework for (P2) and (P3), we view the metasurface as the agent, and the RF sensing scenario including the surroundings, the RF transceiver, and the objects in the target space are regarded, altogether, as the environment.
\revised{
We consider the state of the metasurface the current {\configmat}, i.e., $\bm C$ and the action of the metasurface as selecting the {\newstate} of a {\mselem} for a {\config}.
Thus, actions of the metasurface determine the elements in {\configmat} $\bm C$.}
Therefore, the next state of the MDP is determined by the current state and the action, and the Markov property is satisfied.
In the following, we describe the components of the MDP framework in detail.

\textbf{State}: In the MDP of the metasurface assisted RF sensing scenarios, the state of the environment is defined as
\beq
\label{equ: state def}
\bm s = (k, n, \bm C),
\eeq
where $k\in[1,K]$ and $n\in[1,N]$ are the row and column indexes for {\configmat} indicating the {\newstate} that the metasurface aims to select.
Besides, $\bm C$ in~(\ref{equ: state def}) denotes the current {\configmat} of the metasurface in state $\bm s$.
The initial state of the MDP framework is denoted by $\bm s_0 = (1,1,\bm C_0)$, where $\bm C_0$ is the {\configmat} of the metasurface whose {\mselem}s are in the first/default {\newstate}.
We refer to the states with indices $(k,n)=(K+1,1)$ as the \emph{terminal states}.
When the terminal states are reached, all the {\newstate}s in the {\configmat} have been selected.

\textbf{Action}: In each state $\bm s = (k, n, \bm C)$, the metasurface selects the state of the $n$-th {\mselem} in the $k$-th frame.
The action set of the metasurface in each state can be expressed as $\mathcal A = \{1,...,N_S\}$, where the $j$-th action ($i\in[1,N_S]$) indicates that the metasurface selects the target {\newstate} to be the $j$-th {\newstate}.
In other words, the metasurface sets $\left(\bm C\right)_{k,n} = \hat{\bm o}(a)$, $a\in\mathcal A$.

\begin{figure}[!t] 
\center{\includegraphics[width=0.5\linewidth]{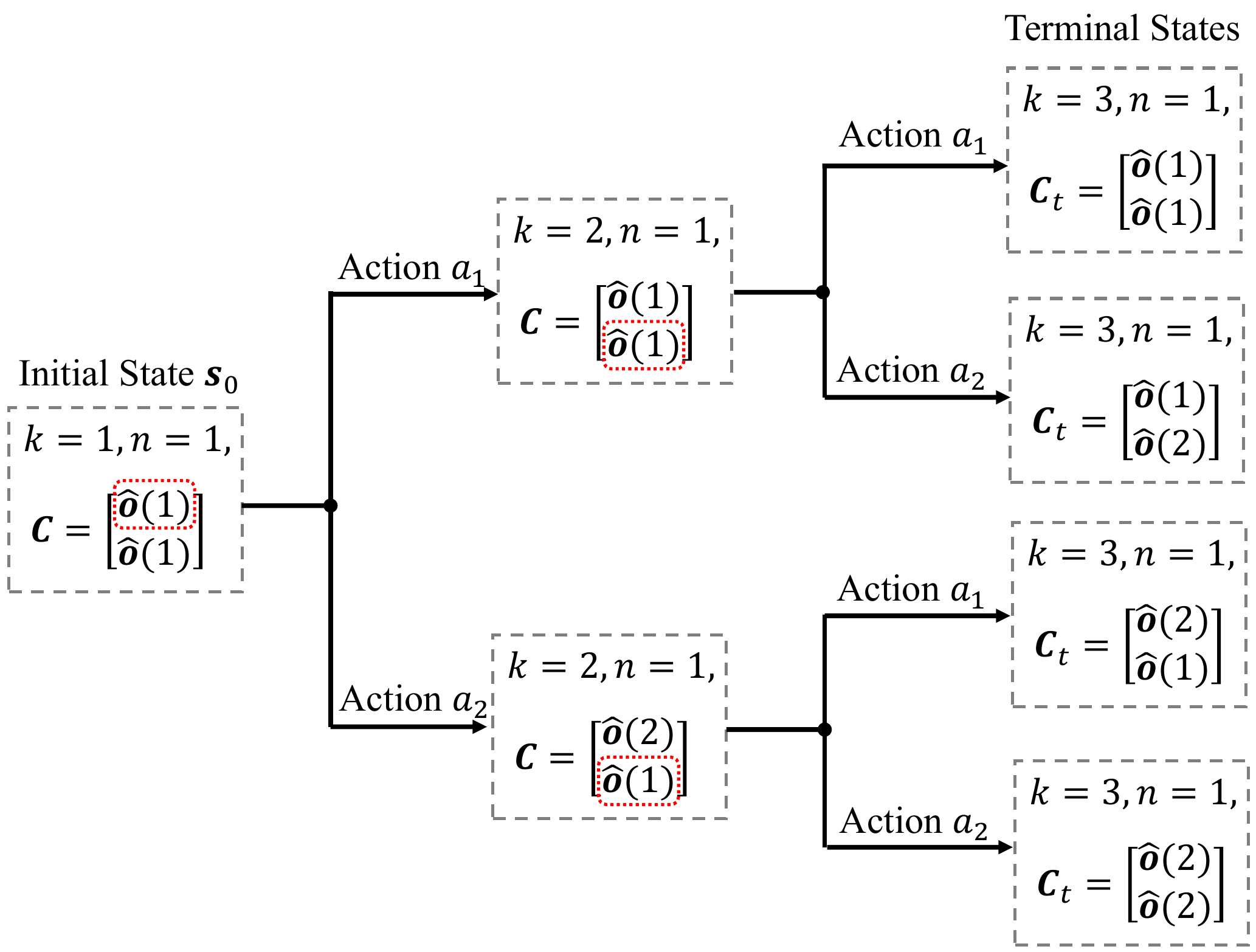}}
	\vspace{-0.5em}
	\setlength{\belowcaptionskip}{-1.em}   
	\caption{Example of the state transition in the formulated MDP, with $K=2$, $N=1$, and $N_S=2$.}
	\vspace{-1em}
	\label{fig: mdp transition}
\end{figure}

\textbf{State Transition Function}:
After the metasurface selects the action, the MDP framework transits into the next state, $\bm s' = (k', n', \bm C' )$, if $(k,n)\neq (K,N)$.
Or, if $(k,n) = (K,N)$, the agent enters the terminal state of the MDP.
For the non-terminal states, the elements of state $\bm s'$ given $\bm s$ and $ a$ can be expressed as follows
\begin{align}
&k' = k+1
, 
\label{equ: transition 2}
\quad n' =  \mathrm{mod}(n+1, N)	+ 1, \\
\label{equ: transition 3}
&\left(\bm C'\right)_{k'',n''}= \begin{cases}
 	\left(\bm C\right)_{k'',n''}, ~\text{if $(k'',n'')\neq(k,n)$},\\
 	\hat{\bm o}(a) ~\text{if $(k'',n'') = (k,n)$},
 \end{cases}\\
 &\hspace{14em}\forall k''\in[1,K], ~n''\in[1,N]. \nonumber
\end{align}

An example of the state transition is illustrated in Fig.~\ref{fig: mdp transition}, where $N_S=2$, $K=2$, and $N=1$. 
In Fig.~\ref{fig: mdp transition}, the red dotted box indicates the element of $\bm C$ that is determined by the action in the current state.
If $(k,n)=(3,1)$, it can be observed that all the {\newstate}s of the {\configmat} have been determined, and the MDP transits into the terminal states, where {\configmat} is denoted by $\bm C_t$.

\textbf{Reward Function}: In general MDP frameworks, the reward is a value obtained by the agent from the environment and quantifies the degree to which the agent's objective has been achieved~\cite{sutton1998reinforcement}.
The reward for the agent is defined as the negative cross-entropy loss of the {\sensfunc} given the {\configmat} determined in the terminal states.
If the terminal state has not been reached, the reward for the state transition is set to be zero.
Specifically, given parameter $\bm w$, the reward in state $\bm s$ is defined as
\beq
\label{equ: mdp reward func}
\mathcal R(\bm s| \bm w) =\begin{cases}
	- L_{\mathrm{CE}}(\bm C_t, \bm w),&\quad \text{if $\bm s$ is a terminal state},\\
	0,&\quad \text{otherwise}.
\end{cases}
\eeq

In the formulated MDP, the metasurface aims for obtaining an optimal policy to obtain the maximum reward in the terminal states.
To be specific, the policy of the agent is a mapping from the state set to the available action set, i.e., $\bm \pi: \mathcal S\rightarrow \mathcal A$.
To define the optimal policy $\bm \pi^*$, we first define the \emph{state-value function} given policy $\bm \pi$ and parameter vector $\bm w$, which indicates the accumulated reward of the agent via a certain state. 
Based on~(\ref{equ: mdp reward func}), the state-value function can be expressed~as
\beq
\label{equ: state value function}
V(\bm s|\bm \pi, \bm w) = \begin{cases}
 - L_{\mathrm{CE}}(\bm C, \bm w),~ \text{if $\bm s$ is a terminal state,}\\
V(\bm s'|\bm \pi,\bm w)|_{\bm s'=\mathcal T(\bm s, \bm\pi(\bm s))}, ~ \text{otherwise,}
\end{cases}
\eeq
The state-value function for $\bm \pi$ in state $\bm s$ indicates the accumulated rewards of the agent after state $\bm s$. Based on~(\ref{equ: state value function}), the state-value function for the initial state can be expressed as
\begin{align}
V(\bm s_0|\bm \pi, \bm  w) &= - L_{\mathrm{CE}}(\bm C^{\bm \pi}_t, \bm  w),
\end{align}
where $\bm C^{\bm \pi}_t$ denotes the terminal state of the metasurface adopting policy~$\bm \pi$.

Therefore, given parameter vector $\bm  w$, the optimal policy of the agent in the MDP framework is given by
\beq
\label{equ: optimal policy given theta}
\bm \pi^*(\bm  w) = \arg\max_{\bm \pi} V(\bm s_0|\bm \pi, \bm  w) \iff \arg\min_{\bm C} L_{\mathrm{CE}}(\bm C, \bm  w).
\eeq
In~(\ref{equ: optimal policy given theta}), it can be observed that finding the optimal policy of the agent in the formulated MDP framework is equivalent to solving the optimal {\configmat} for~(P3).
Besides, solving~(P2) is equivalent to solving the optimal $\bm w$ given the policy $\bm \pi$.

\subsection{Progressing Reward Policy Gradient Algorithm}
To jointly solve (P2) and (P3) under the formulated MDP framework, we propose a novel PRPG algorithm.
The proposed algorithm can be divided into two phase, i.e., the \emph{action selection phase} and the \emph{training phase}, which proceed iteratively.

\subsubsection{Action Selection Process}
\label{sssec: action selection process}
In the proposed algorithm, the agent, i.e., the metasurface, starts from the initial state $\bm s_0$ and adopts the policy for selecting action in each state until reaching the terminal state.
To select the current action in each state, the metasurface use policy $\bm\pi$ that maps the current state to a probability vector.
To be specific, for a given state $\bm s$, the policy results in an $N_S$-dimensional probability vector denoted by $\bm \pi(\bm s|\bm  w)$, which we refer to as the \emph{policy function}.
The $i$-th element of $\bm \pi(\bm s|\bm  w)$ $(i\in[1,N_S])$, i.e., $\pi_{i}(\bm s|\bm  w)$, is in range $[0,1]$ and denotes the probability of selecting the action $a_i$ in state $\bm s$.
Besides, $\bm \pi(\bm s|\bm  w)$ $(i\in[1,N_S])$ satisfies $\sum_{i=1}^{N_S}\pi_{i}(\bm s|\bm  w) = 1$.

However, since the state contains the current {\configmat} that contains $K \cdot N \cdot N_S$ binary variables, the agent faces a large state space, and the policy function is hard to be modeled by using simple functions.
To handle this issue, we adopt a neural network to model the policy function as neural networks are a powerful tool to handle large state space~\cite{zappone2019Model}.
The adopted neural network is referred to as the policy network, and we train the policy network by using the policy gradient algorithm~\cite{Volodymyr2015Human}.
Specifically, the policy network is denoted by $\bm \pi^{\bm \theta}(\bm s|\bm  w)$, where $\bm \theta$ denotes the parameters of the policy network and comprises the connection weights and the biases of the activation functions in the neural network.

\begin{figure}[!t] 
\center{\includegraphics[width=0.6\linewidth]{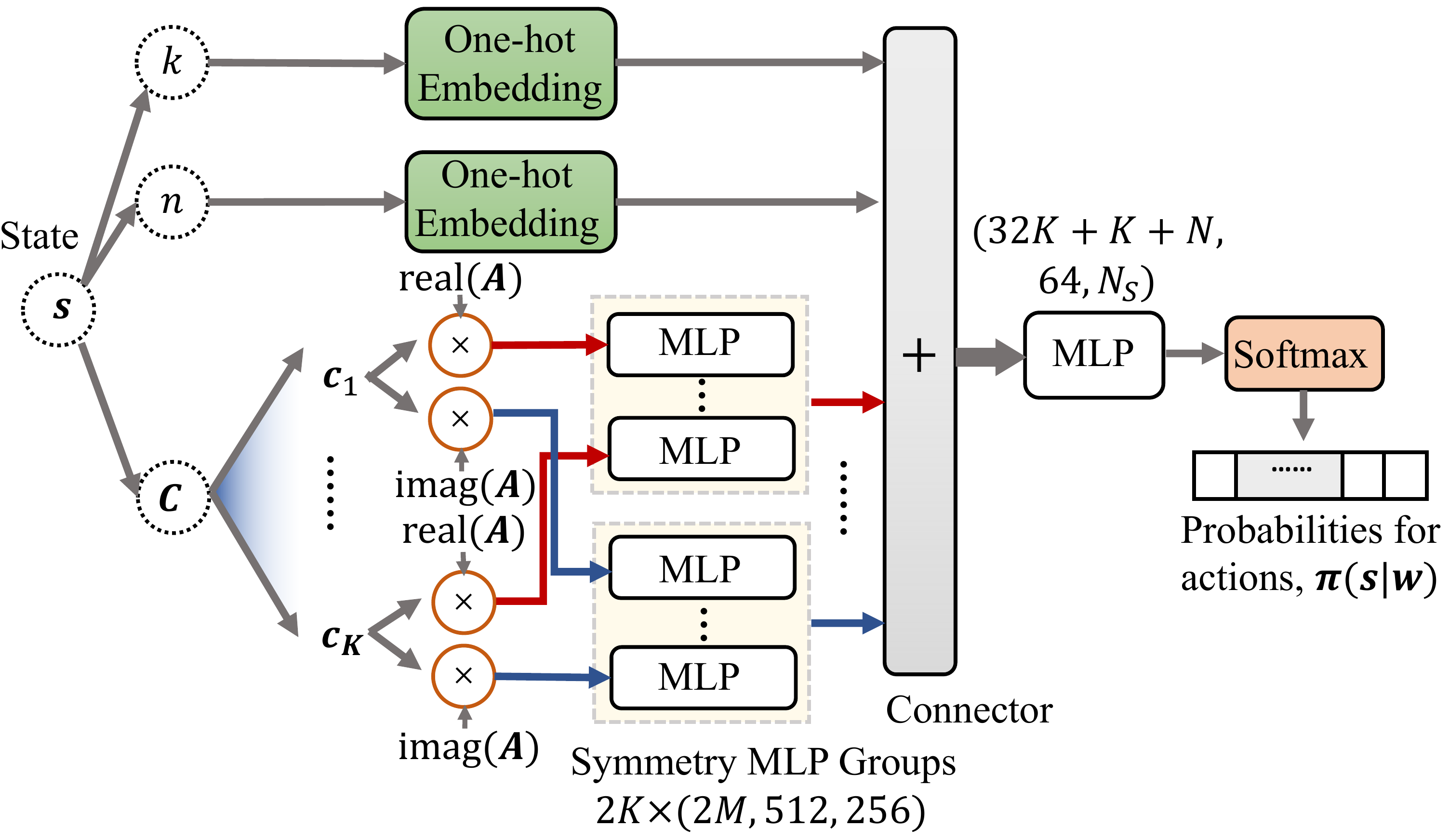}}
	\vspace{-1em}
	\setlength{\belowcaptionskip}{-1.em}   
	\caption{Network structure of the policy network used in the proposed algorithm.}
	\vspace{-1.em}
	\label{fig: policy net}
\end{figure}

The structure of the policy network is shown in Fig.~\ref{fig: policy net}.
In state $\bm s$, $k$ and $n$ are embedded as a $K$-dimensional and an $N$-dimensional vectors, respectively, where the $k$-th and $n$-th elements in the vectors are ones and the other elements are zeros.
Specifically, we refer to the resulted vectors as the \emph{one-hot} vectors.
As for $\bm C$, since the RF sensing for the target space is determined by $\bm C\bm A$ as shown by (\ref{equ: received signal matrix form}), we first divide $\bm C$ to its real and imaginary parts and right-multiply them by the real and imaginary parts of $\bm A$, respectively.
Then, driven by the concept of model-based learning~\cite{Zappone2019Wireless},  we process the result, i.e., $\bm CA$, by multi-layer perceptrons~(MLPs).
Besides, since the $K$ {\config}s are symmetric in their physical meaning and changing their order does not impact the sensing performance, the MLPs that extract feature vectors from $\bm c_1$ to $\bm c_K$ need to be symmetric.
This can be achieved by utilizing two \emph{symmetric MLP groups}, each containing $K$ MLPs with shared parameters.
This significantly reduces the number of parameters and thus facilitates the training of the policy network.
The sizes of the MLPs are labeled in Fig.~\ref{fig: policy net}.
For example, $(2M, 512, 256)$ indicates that each MLP in a symmetric group that has three layers whose sizes are $2M$, $512$, and $256$, respectively.
Then, the one-hot vectors and the $2K$ extracted feature vectors are connected and input to the final MLP.
The result of the final MLP is fed into the \emph{softmax} layer which produces an $N_S$-dimensional vector indicating the probability of selecting the $N_S$ actions.

\subsubsection{Training Process}
The purpose of the training process is two-fold: \emph{(a)} To make the policy network improves the current policy in action selection based on~(\ref{equ: optimal policy given theta}). 
\emph{(b)} To make the {\sensfunc} incur lower cross-entropy loss.
Accordingly, the training process consists of two parts, i.e., training of the policy network and training of the sensing network.
In the training of the policy network, we adopt the policy gradient method~\cite{sutton1998reinforcement}.
Besides, the training of the sensing network results in that the rewards for the terminal states progress during the training of the policy.
Due to these characteristics, the proposed algorithm is named as \emph{progressing reward policy gradient} algorithm.

\textbf{Training of the Policy Network}:
To collect the training data for the policy network, a \emph{replay buffer} is adopted in order to store the experiences of the agent during state transitions.
The replay buffer of the agent is denoted by $\mathcal B = \{\bm e\}$.
The stored experience in the replay buffer is given by $
\bm e = (\bm s, a).
$
It is worth noting that, differently from the replay buffer in traditional deep reinforcement learning algorithms~\cite{Volodymyr2015Human}, the experience in the replay buffer does not record the reward obtained during the state transitions.
This is because the rewards are determined by the current {\sensfunc}, which changes as $\bm w$ being updated.
Thus, we propose that the rewards are calculated when the training process is invoked, instead of being recorded in the replay buffer.

We define a \emph{training epoch}~(or epoch in short) as the state transition process from the initial state to a terminal state.
The experience of the agent within an epoch is stored into the replay buffer and used for training, which is discarded after being used.
Based on the policy gradient theorem\cite{sutton1998reinforcement}, in the training process, the gradient of $V(\bm s_0|\bm \pi, \bm w)$ with respect to $\bm \theta$ satisfies
\beq
\label{equ: value gradient}
\nabla_{\bm \theta} V(\bm s_0 |\bm \pi, \bm w) \!\propto\! \mathbb E_{\mathcal B, \bm \pi^{\bm \theta}}\!\left[
	V(\mathcal T(S_t, A_t)|\bm \theta, \bm w)\frac{\nabla_{\bm\theta} \pi^{\bm \theta}_{A_t}(S_t|\bm w)}{\pi^{\bm \theta}_{A_t}(S_t|\bm w)}
	\right],
\eeq
where $(S_t, A_t)\in\mathcal B$ are the samples of the state and action in the replay buffer of an agent following policy $\bm \pi^{\bm \theta}$, and $Q(S_t, A_t|\bm \theta, \bm w)$ denotes the reward for the agent after selecting the action $A_t$ in $S_t$ and then following~$\bm \pi^{\bm \theta}$.

To calculate the gradient in~(\ref{equ: value gradient}), the rewards for the agent in~(\ref{equ: mdp reward func}) need to be calculated.
If $\bm s$ is a terminal state, the reward $R(\bm s|\bm w)$ is calculated by using the Monte Carlo methods~\cite{Rubinstein2008Simulation}, i.e.,
\begin{align}
\label{equ: calculate r by monte carlo}
R(\bm s|\bm w)  &= 
- \sum_{\bm \nu\in\mathcal V}
\sum_{i=1}^{N_{\mathrm{mc}}} 
\Big(\sum_{m=1}^M  p_{m}(\bm \nu) \ln(\hat{p}_{m}) +(1-p_{m}(\bm \nu)) \ln(1-\hat{p}_{m}) \Big)
\Big|_{\hat{\bm p} = \bm f^{\bm w}(\nGamma \bm \nu + \tilde{\bm \sigma}_i)}. 
\end{align}
Otherwise, $R(\bm s|{\bm w}) = 0$.
In~(\ref{equ: calculate r by monte carlo}), $N_\mathrm{mc}$ indicates the number of sampled noise vectors, and $\tilde{\bm \sigma}_i$ is the $i$-th sampled noise vector.
As the rewards in the non-terminal states are zero, $V(\mathcal T(S_t, A_t)|\bm \theta, \bm w)$ is equal to the reward at the final state for $S_t$, $A_t$, and policy $\bm \pi^{\bm \theta}$.

\begin{figure}[!t] 
\center{\includegraphics[width=0.6\linewidth]{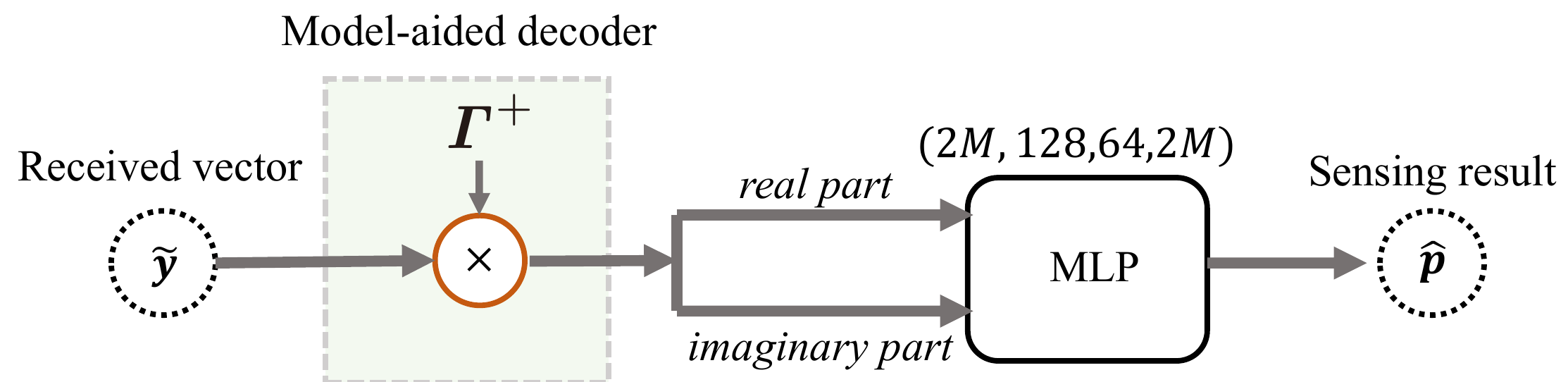}}
	\vspace{-1.em}
	\setlength{\belowcaptionskip}{-1.em}   
	\caption{Sensing network of the metasurface.}
	\vspace{-1.em}
	\label{fig: image net}
\end{figure}

Specifically, in~(\ref{equ: calculate r by monte carlo}), $\hat{\bm p}$ is generated by the sensing network, which is shown in Fig.~\ref{fig: image net}.
The sensing network consists of two parts, i.e., the \emph{model-aided decoder} and an MLP.
Firstly, the received vector is left-multiplied by the \emph{pseudo inverse} of $\nGamma$, which is denoted by $\nGamma^+$ and can be calculated based on~\cite{bailyn2004generalized}.
According to the least-square method~\cite{Boyd_CONVEX}, the model-aided decoder, i.e., $\hat{\bm \nu} = \nGamma^{+} \bm y$, is the optimal linear decoder that results in the minimum mean square error~(MSE) for the actual reflection vector $\bm \nu$, and thus can potentially increase the sensing accuracy of the sensing network.
Then, $\hat{\bm \nu}$ is fed into a fully-connected MLP, which reconstructs the probability vector~$\hat{\bm p}$.

In each process, $\bm\theta$ is updated as follows
\begin{equation}
\label{equ: update of theta}
\bm\theta = \bm \theta +\alpha\cdot{\mathbb E}_{\bm e \in \mathcal B}\left[
    V(\mathcal T(S_t, A_t)|\bm \theta, \bm w)\frac{\nabla_{\bm\theta} \pi^{\bm \theta}_{A_t}(S_t|\bm w)}{\pi^{\bm \theta}_{A_t}(S_t|\bm w)}\right],
\end{equation}
where the gradient $\nabla_{\bm\theta} \pi^{\bm \theta}_{A_t}(S_t|\bm w)$ is calculated by using the back-propagation algorithm~\cite{pineda1987generalization}, and $\alpha$ denotes the training rate.

\textbf{Training of the Sensing Network}:
After updating $\bm \theta$, the training of the sensing network is executed.
The calculated rewards from (\ref{equ: calculate r by monte carlo}) are used to train the sensing network which reduces the cross-entropy loss.
To be specific, the loss function used to train the sensing network can be expressed as follows, which is in accordance with the objective function in the optimization problem~(P2), i.e., 
\beq
\label{equ: sensing function loss}
\mathcal L_{I}(\bm w) = {\mathbb E}_{(\bm s, a)\in\mathcal B}[R(\bm s|{\bm w})].
\eeq
In each training process, $\bm w$ is updated by
\beq
\label{equ: w update}
\bm w = \bm w + \alpha\nabla_{\bm w} \mathcal L_{I}(\bm w),
\eeq
where the gradient $\nabla_{\bm w} \mathcal L_{I}(\bm w)$ is calculated by using the back-propagation algorithm.

In summary, the proposed PRPG algorithm is summarized in Algorithm~\ref{alg: proposed algorithm}.
\begin{algorithm}[!t]
\small
	\label{alg: proposed algorithm}
  \caption{Proposed PRPG Algorithm}
  \begin{algorithmic}  
  \Require{
	Random initial network parameter vectors $\bm \theta$ and $\bm w$\;
	Empty replay buffer $\mathcal B = \emptyset$\;
	Maximum number of training epochs $N_{\mathrm{ep}}$\;
	Set of reflection coefficient vectors $\mathcal V$\;
	Number of Monte Carlo samples for noise $N_{\mathrm{mc}}$\;
	Initial learning rate $\alpha_0$\;
	Maximum number of training epochs $N_{\mathrm{ep}}$
	}
  \Ensure{
  Optimized sensing network parameter vector $\bm w^*$ and the optimized policy network parameter $\bm \theta^*$.
  }\\
 
   \For{$n_{\mathrm{ep}} = 1$ to $N_{\mathrm{ep}}$}
   {   
		Set the current state to be the initial state, i.e., $\bm s = \bm s_0$\;
		
		\emph{\# Action selection phase}
		
   		\While{$\bm s$ is not a terminal state}
   		{
   			Select the {\newstate} of the $n$-th {\mselem} in the $k$-th frame following the probability distribution given by $\bm \pi^{\bm\theta}(\bm s|\bm w)$.	
   			
   			Set action $a$ as the selected {\newstate}, and enter into the transited state $\bm s' = \mathcal T(\bm s, a)$\;
   			
   			Store experience $\bm e = (\bm s, a)$ into replay buffer~$\mathcal B$\;
	   }
	   
	  \emph{\# Training phase}
	   
	   Collect all the experiences from $\mathcal B$, and calculate the reward for each sampled experience by using~(\ref{equ: calculate r by monte carlo})\;	   
	   
	   Update parameter $\bm \theta$ and $\bm w$ by (\ref{equ: update of theta}) and~(\ref{equ: w update}), respectively, where the learning rate $\alpha = \frac{\alpha_0}{1+n_{\mathrm{ep}}\cdot 10^{-3}}$\;
   }
   \end{algorithmic}
\end{algorithm}

\textbf{\revised{Remark:}}
\revised{
    Using the proposed deep reinforcement learning technique enables our proposed algorithm to handle the complicated cases where multiple metasurfaces exist.
    Specifically, when the multiple metasurfaces are on the same plane, they can be considered as a whole, and thus the channel model in (\ref{equ: received signal}) needs no changes.
    When the multiple metasurfaces are on different planes, the channel model needs to be modified to adapt to the correlation between different metasurfaces, which is left for future work.
    Nevertheless, since the problem formulation and the proposed algorithm are independent of the specific channel model, the proposed problem formulation and algorithm can also be adopted for the scenarios to optimize the sensing performance of the general RF sensing scenarios with multiple metasurfaces.}

\section{Algorithm Analysis}
\label{algorithm analysis}

In this section, we analyze the computational complexity and the convergence of the proposed algorithm in Subsections~A and~B, respectively.
In addition, in Subsection~C, we derive a non-trivial lower-bound for the sensing accuracy based on an upper-bound for the cross-entropy loss given a {\configmat}.

\subsection{Computational Complexity}
Since the PRPG algorithm consists of two main phases, i.e., the action selection phase and the training phase, we analyze their respective computational complexities.
The computational complexities are analyzed with regard to the number of {\config}s, $K$, the number of {\mselem}s, $N$, the number of available {\newstate}, $N_S$, and the number of space grids, $M$.

\subsubsection{Complexity of the Action Selection Phase}
In the proposed algorithm, the computationally most expensive part is the estimation of the action probabilities of the policy network.
For each action selection phase, the computational complexity is given in Theorem~\ref{theo: action determine complexity}.

\begin{theorem} \label{theo: action determine complexity}
\emph{(Computational Complexity of the Action Selection Phase)}
In the PRPG algorithm, for the agent in each state, the complexity to calculate the action probabilities and determine the action is $\bigO( KNN_SM)$.
\end{theorem}

\begin{proof}
See Appendix~\ref{appx: theorem 1}.
\end{proof}

\subsubsection{{Complexity of the Training Process}}
The computational complexity of~(\ref{equ: calculate r by monte carlo}) is provided in Lemma~\ref{lemma: reward calculation}.
\begin{lemma}
\label{lemma: reward calculation}
	The computational complexity of the reward calculation in~(\ref{equ: calculate r by monte carlo}) is 
	\beq
	\bigO\!\left(KNN_SM+K^2M+M^2  \right). \nonumber
	\eeq
\end{lemma}
\begin{proof}
See Appendix~\ref{appx: proof of lemma 1}	
\end{proof}

The computational complexities of training the policy network and the sensing network are given in Lemma~\ref{lemma: training calculation complexity}.
\begin{lemma}
\label{lemma: training calculation complexity}
	After calculating the rewards, the complexity of the training the sensing network and the policy network are $\bigO(M^2)$ and $\bigO(N_S(K+N+M))$, respectively.
	If a single MLP is used to substitute the symmetric MLP group, the computational complexity of training the policy network is $\bigO(KMN_S+NN_S)$.
\end{lemma}
\begin{proof}
See Appendix~\ref{appx: proof of lemma 2}.	
\end{proof}

It can be observed from Lemma~\ref{lemma: training calculation complexity} that using a symmetric MLP group instead of a single large MLP in the policy network can reduce the complexity of the training process.

Based on Lemmas~\ref{lemma: reward calculation} and~\ref{lemma: training calculation complexity}, the total computational complexity of each training process is provided in Theorem~\ref{theo: training complexity}.
\begin{theorem}\label{theo: training complexity}
\emph{(Computational Complexity of the Training Process)}
The computational complexity of each training phase of the PRPG algorithm is $\bigO\!\left(KNN_SM+K^2M+M^2 \right).$\end{theorem}
\begin{proof}
See Appendix~\ref{appx: proof of theorem 2}.	
\end{proof}

\subsection{Convergence Analysis}
\label{sec: convergence analysis}

The detailed convergence analysis of the PRPG algorithm is based on the convergence analysis of the block stochastic gradient~(BSG) algorithm. 
We denote $\bm w$ by $\bm x_1$ and denote $\bm \theta$ by $\bm x_2$, and thus the objective function in~(P1) can be denoted by $F(\bm x_1, \bm x_2) = L_{\mathrm{CE}}(\bm C^{\bm \pi^{\bm\theta}}_t, \bm w)$, where $\bm C^{\bm \pi^{\bm\theta}}_t$ indicates the {\configmat} in the terminal state for the metasurface with policy $\bm \pi^{\bm \theta}$.
Based on~\cite{Xu2015Block}, a BSG algorithm for solving~(P1) is formulated as Algorithm~\ref{alg: bsg algorithm}, whose convergence analysis can be given by Lemma~\ref{lem: bsg converge}.
\begin{lemma}
\label{lem: bsg converge}
Algorithm~\ref{alg: bsg algorithm} converges to a locally optimal $\bm x^*_1$ and $\bm x^*_2$ as the number of iterations $N_{\mathrm{itr}}\rightarrow\infty$,  given that the following conditions are satisfied:
\begin{enumerate}
\item There exist a constant $c$ and a constant $\varepsilon$ such that, for each iteration indexed by $j$, the inequalities $\|\mathbb E[\tilde{\bm g}_{i}^j - \nabla_{\bm x_i} F(\bm x_1, \bm x_2)] \|_2 \leq c \cdot \max_i(\alpha^j_{i})$ and $\mathbb E[\|\tilde{\bm g}_{i}^j - \nabla_{\bm x_i} F(\bm x_1, \bm x_2)\|^2] \leq \varepsilon^2$, $i=1,2$ are fulfilled.
\item There exists a uniform Lipschitz constant $\varrho>0$ such that 
\beq
\hspace{-1.8em}\sum_{i=1,2}\!\|\nabla_{\bm x_i}F(\bm x_1, \bm x_2)\! -\! \nabla_{\bm x_i}F(\bm x_1', \bm x_2')\|_2^2 \! \leq \! \varrho^2 \! \sum_{i=1,2}\|\bm x_i - \bm x_i'\|^2_2. \nonumber
\eeq
\item There exists a constant $\psi$ such that $\mathbb E[\|\bm x_1^{j}\|_2^2 + \|\bm x_2^{j}\|_2^2]\leq \psi^2$, $\forall j$.
\end{enumerate}
\end{lemma}
\begin{proof}
Please refer to Corollary 2.12 in~\cite{Xu2015Block}, where the assumptions required in Corollary~2.12 in~\cite{Xu2015Block} are equivalent to the three conditions in Lemma~\ref{lem: bsg converge}.
\end{proof}

Comparing Algorithms~\ref{alg: proposed algorithm} and~\ref{alg: bsg algorithm}, we can observe that the only difference between the two algorithms is in the functions for updating parameters.
Nevertheless, solving the minimization problem (\ref{equ: bsg para update func}), we can derive that  (\ref{equ: bsg para update func}) is equivalent to that 
\beq
\bm x_i^j = \bm x_i^{j-1} - \alpha_i^{j}\tilde{\bm g}_i^j.
\eeq
As the learning rate sequence $\{\alpha_i^j\}_j$ in Algorithm~\ref{alg: bsg algorithm} can be arbitrarily selected, the parameter update of Algorithms~\ref{alg: proposed algorithm} and~\ref{alg: bsg algorithm} are essentially equivalent.
In this regard, the proposed PRPG algorithm can be categorized as an BSG algorithm, whose convergence analysis follows Lemma~\ref{lem: bsg converge}.

However, since neural networks are encompassed in the {\sensfunc} and the policy function, the conditions in Lemma~\ref{lem: bsg converge} are hard to be proven theoretically.
Therefore, in additional to the theoretical analyses provided above, we also analyze the convergence through practical simulations in Section~\ref{sec: simulation result}.

\revised{Moreover, the obtained solution by the proposed deep learning algorithm is a locally optimal solution of (P1).
As shown in Algorithm~\ref{alg: proposed algorithm}, we iteratively solve~(P2) and~(P3) by updating $\bm \theta$ using (\ref{equ: update of theta}) and updating $\bm w$ using (\ref{equ: w update}), respectively.
Based on the Q-learning algorithm~\cite{sutton1998reinforcement}, updating $\bm \theta$ with the aim to maximize the total reward is equivalent to finding $\bm C$ minimizing $L_{\mathrm{CE}}$ given $\bm w$.
Besides, it can be observed that updating $\bm w$ directly minimizes $L_{\mathrm{CE}}$ given $\bm C$.
When the iteration terminates, updating the variables of $\bm C$ or $\bm w$ will not lead to a lower objective function value, i.e., the cross-entropy loss.
Therefore, the solution obtained by the proposed Algorithm~\ref{alg: proposed algorithm} is a locally optimal solution of the original problem (P1).}

\begin{algorithm}[!t]
\small
	\label{alg: bsg algorithm}
  \caption{BSG algorithm for solving~(P1)}
  \begin{algorithmic}  
  \Require{
	Starting point $\bm x_i^{0}$, $i=1,2$\;
	Learning rate sequence $\{\alpha_{i}^{j}; i=1,2\}_{j=1,2,...}$\;
	Maximum number of iterations $N_{\mathrm{itr}}$\;
	Monte Carlo sampling size of the random noise $N_{\mathrm{mc}}$.
	}
  \Ensure{
  Optimized $\bm x_1^*$ and $\bm x_2^*$ for (P1).
  }\\
 
   \For{$j=1,2,...,N_{\mathrm{itr}}$}
   {   

		\For{$i=1,2$}
		{
			Compute sample gradient for the $\bm w$ in the $j$-th iteration by
			$
			\tilde{\bm g}_{i}^{j} =
			\nabla_{\bm x_i} F(\bm x_{<i}^{j}, \bm x_{\geq i}^{(j-1)})
			$
			
			Update parameter $\bm x_i$ by
			\beq
			\label{equ: bsg para update func}
			\bm x^{j}_i = \arg\min_{\bm x_i} 
				(\tilde{\bm g}_{i}^{j})^T(\bm x_i - \bm x_i^{j-1}) 
				+ \frac{1}{2\alpha_{i}^{j}} 
					\| \bm x_i - \bm x_i^{j-1}\|^2_2. \nonumber
			\eeq
		}		
   }\\
   
Output $(\bm x_1^{N_{\mathrm{itr}}}, \bm x_2^{N_{\mathrm{itr}}})$ as $(\bm x_1^*, \bm x_2^*)$\;
   \end{algorithmic}
\end{algorithm}

\subsection{Lower Bound for Sensing Accuracy}
\label{sec: analytical results}

In this section, we compute a lower-bound for the sensing accuracy in~(P2) given {\configmat} $\bm C$.
To derive a lower bound, we assume that the {\sensfunc} maps the received RF signals to the sensing results by using an optimal linear decoder and a threshold judging process.
In the following, we first provide the detection criterion for sensing, and then derive a lower-bound for sensing accuracy by leveraging an upper-bound for the cross-entropy loss.

\subsubsection{Detection Criterion for Sensing}
The reconstructed reflection coefficient vector from the linear decoder can be expressed as
\beq
\label{equ: initial analysis equation}
\hat{\bm \nu} = \nGamma^{+}\tilde{\bm y} = \nGamma^+\nGamma\bm \nu + \nGamma^+ \tilde{\bm \sigma}.
\eeq

Based~(\ref{equ: initial analysis equation}), we analyze the probability distribution of the random variable $\hat{\nu}_m$, i.e., the $m$-th element of $\hat{\bm \nu}$.
We denote the $m$-th row vectors of $\nGamma^+$ and $\nGamma^+ \nGamma$ as $\bm \gamma_m$ and $\bm \xi_m$, respectively.
Then, $\hat{\nu}_m = \bm\xi_m\bm\nu+\bm\gamma_m\tilde{\bm\sigma}$.
The emptiness of the space grids other than the $m$-th space grid is modeled by the vector $\bm q_{-m}$, where $q_{-m, m'}=0$ and $1$ indicate that the $m'$ space grid is empty and nonempty, respectively,~($m'\in[1,M],~m'\neq m$).
When the $m$-th space grid is empty~(or nonempty), we denote the probability density functions~(PDFs) of the real and imaginary parts of $\hat{\nu}_m$, i.e., $\hat{\nu}_{R,m}$ and $\hat{\nu}_{I,m}$, by ${\mathcal P}^0_{R,i}(x)$ and ${\mathcal P}^0_{I,i}(x)$~(or ${\mathcal P}^1_{R,i}(x)$ and ${\mathcal P}^1_{I,i}(x)$), respectively.

We judge the emptiness of the $m$-th space grid according to the sum of $\hat{\nu}_{R,m}$ and $\hat{\nu}_{I,m}$, i.e., $\mu_m = \hat{\nu}_{R,m} + \hat{\nu}_{I,m}$.
When the $m$-th space grid is empty, given $\bm q_{-m}$, the sum of $\hat{\nu}_{R,m}$ and $\hat{\nu}_{I,m}$, i.e., $\mu_m$, follows a normal distribution, i.e., $\mu_m \sim \mathcal N(0, \epsilon_m^0(\bm q_{-m}))$, where 
\begin{align}
\label{equ: var when empty}
\epsilon^0_m(\bm q_{-m}) = & \sum_{\substack{m'\neq m,\\ m'\in\mathcal M}}  q_{-m,m'} \cdot \epsilon_{\reflect, m'} \cdot (\|{\bm \xi}_{R,m'}\|^2 + \|{\bm \xi}_{I,m'}\|^2) \nonumber \\
&+ \sum_{m'\in\mathcal M} \epsilon \cdot (\|{\bm \gamma}_{R,m'}\|^2 + \|{\bm \gamma}_{I,m'}\|^2).
\end{align}
Here, $\mathcal M$ is the set of indexes of $M$ space grids, and subscripts $R$ and $I$ indicate the real and imaginary parts of a vector, respectively.
The first summation term in~(\ref{equ: var when empty}) corresponds to the variance due to the reflection coefficients at the space grids other than the $m$-th space grid, and the second summation term in~(\ref{equ: var when empty}) corresponds to the variance due to the noise at the Rx unit.

On the other hand, when the $q$-th space grid is nonempty, the variance due to reflection coefficient of the $m$-th space grid needs to be added.
Denote the variance of the reflection coefficient of the $m$-th space grid by $\epsilon_{\reflect, m}$, and the variance of $\mu_m$ can be expressed as
\beq 
\label{equ: var when nonempty}
\epsilon^1_m(\bm q_{-m}) = 
    \epsilon^0_m(\bm q_{-m}) 
        + 
    \epsilon_{\reflect, m} \cdot (\|{\bm \xi}_{R,m}\|^2 + \|{\bm \xi}_{I,m}\|^2).
\eeq

Given the emptiness of the $m$-th space grid, the PDF of $\mu_m$ can be written as follows
\begin{align}
    \label{equ: theo sum pdf func}
    \mathcal P^i_{m}(x) =\!  
    \sum_{\bm q_{-m}\in\mathcal Q_{-m}} 
    \! P_m(\bm q_{-m})
    {\mathcal P}_{norm}(x; 0, \epsilon_m^i(\bm q_{-m})), ~ i=0,1
\end{align}
where $\mathcal Q_{-m}$ indicates the set of all possible $\bm q_{-m}$,
$ {\mathcal P}_{norm}(x; 0, \epsilon^i_m(\bm q_{-m}))$~($i=0,1$) denotes the PDF of a normal distribution with zero mean and variance $\epsilon^i_m(\bm q_{-m})$,
and $P_m(\bm q_{-m})$ denotes the probability for the existence indicated by $\bm q_{-m}$ to be true, i.e.,
\beq
P_m(\bm q_{-m}) = \prod_{m'\neq m, m'\in\mathcal M} Pr_{m'}(q_{-m,m'}).
\eeq
Here, $Pr_{m'}(x)$ with $x$ being $0$ and $1$ indicates the probabilities that the $m'$-th space grid are empty and nonempty, respectively.

We use the difference between $\mathcal P_m^1(\bm q_{-m})$ and $\mathcal P_m^0(\bm q_{-m})$ as the \emph{judgement variable} to determine whether the $m$-th space grid is empty or not.
To facilitate the analysis, we adopt the \emph{log-sum} as a substitute for the sum in (\ref{equ: theo sum pdf func}).
Therefore, the judgement variable can be calculated~as
\begin{align}
\label{equ: lambda}
\tau_m \!= \!
    &\sum_{\bm q_{-m}\in\mathcal Q_{-m}} 
        \ln\left( p_m(\bm q_{-m}) 
        {\mathcal P}_{norm}(x; 0, \epsilon^1_m(\bm q_{-m})) \right) \\
    & - \sum_{\bm q_{-m}\in\mathcal Q_{-m}} \!
        \ln\left(p_m(\bm q_{-m}) 
        {\mathcal P}_{norm}(x; 0, \epsilon^0_m(\bm q_{-m}))\right). \nonumber
\end{align}

It can be observed from~(\ref{equ: lambda}) that $\tau_m$ increases as $\mathcal P_m^1(\mu_m)$ increases, and that it decreases as $\mathcal P_m^0(\mu_m)$ increases.
Therefore, we can judge the emptiness of the $m$-th space grid through the value of $\tau_m$.
Specifically, the sensing result of the $m$-th space grid is determined by comparing the judging variable $\tau_m$ with the \emph{judging threshold}, which is denoted by $\rho_m$.
If $\tau_m \leq \rho_m$, the sensing result of the $m$-th space grid is ``empty'', which is denoted by the hypothesis $\mathcal H_0$.
Otherwise, if $\tau_m > \rho_m$, the sensing result is ``non-empty'', which is denoted by the hypothesis $\mathcal H_1$.
After simplifying~(\ref{equ: lambda}), the \emph{detection criterion} for $\mathcal H_0$ and $\mathcal H_1$ can be expressed as
\begin{align}
\label{equ: judging function original}
\tau_m =  
    \mu_m^2 \sum_{\bm q_{-m}\in Q_{-m}} \frac{\epsilon_m^1(\bm q_{-m})- \epsilon_m^0(\bm q_{-m})}{2\epsilon_m^1(\bm q_{-m})\epsilon_m^0(\bm q_{-m})}  - \frac{1}{2}\sum_{\bm q_{-m}\in\mathcal Q_{-m}} \ln\left(\frac{\epsilon_m^1(\bm q_{-m})}{\epsilon_m^0(\bm q_{-m})}\right)
        \LRT{\mathcal H_1}{\mathcal H_0}
    \rho_m.
    \end{align}
Since $\mu_m^2 > 0$, the range of $\rho_m$ can be expressed $[-\frac{1}{2}\sum_{\bm q_{-m}\in\mathcal Q_{-m}} \ln(\frac{\epsilon_m^1(\bm q_{-m})}{\epsilon_m^0(\bm q_{-m})}), \infty]$.

\vspace{0.5em}
\subsubsection{Upper Bound of Cross Entropy Loss}
We analyze the cross-entropy loss incurred by the detection criterion~in (\ref{equ: judging function original}), which can be considered as a non-trivial upper-bound for the cross-entropy loss defined in~(\ref{equ: def of sensing loss}).
As the sensing result given by~(\ref{equ: judging function original}) is either $0$ or $1$,  if the sensing result is accurate, the incurred cross-entropy loss will be $-\ln(1)=0$; otherwise, the incurred cross-entropy loss will be $-\ln(0) \rightarrow \infty$.
In practice, the cross-entropy loss due to an inaccurate sensing result is bounded by a large number $C_{\mathrm{In0}}$.
Given $\mathcal H_0$~(or $\mathcal H_1$) being true, the probability for the sensing result to be inaccurate is the probability of $\tau_m > \rho_m$, i.e., $\Pr\{\tau_m > \rho_m|\mathcal H_0\}$~(or $\Pr\{\tau_m \leq \rho_m|\mathcal H_1\}$).
Denote the probability for an object to be at the $m$-th space grid by $\tilde{p}_{m}$, and the cross-entropy loss of the $m$-th space grid can be calculated as 
\begin{align}
\label{equ: analytical CE loss of m}
L_m = 
    C_{\mathrm{In0}} \cdot (1 - \tilde{p}_{m}) \cdot \Pr\{ \tau_m > \rho_m | \mathcal H_0 \}  + C_{\mathrm{In0}}\cdot \tilde{p}_{m} \cdot \Pr\{ \tau_m \leq \rho_m | \mathcal H_1 \}, 
\end{align}
where $\Pr\{\tau_m > \rho_m|\mathcal H_0\}$ and $\Pr\{\tau_m \leq \rho_m|\mathcal H_1\}$ can be calculated by using Proposition~\ref{prop: threshold judging}.
\begin{proposition}
\label{prop: threshold judging}
The conditional probability for sensing the $m$-th space grid inaccurately can be calculated as follows
\begin{align}
\label{equ: condi prob prop 1}
\Pr\{ \tau_m > \rho_m | &\mathcal H_0\} 
= \Pr\{ \mu^2_m > \hat{\rho}_m | \mathcal H_0\}
= 1
    \! - \! 
    \sum_{\bm q_{-m}\in\mathcal Q_{-m}}
        P_m(\bm q_{-m})
            \!\cdot\!
        \mathrm{erf}\left(\sqrt{ \frac{\hat{\rho}_m}{2\epsilon_m^0(\bm q_{-m})} }\right), 
\end{align}
\begin{align}
\label{equ: condi prob prop 2}
\Pr\{ \tau_m \leq \rho_m | &\mathcal H_1\}  = 
\Pr\{ \mu^2_m \leq \hat{\rho}_m | \mathcal H_1\} 
= \sum_{\bm q_{-m}\in\mathcal Q_{-m}}
        P_m(\bm q_{-m}) 
            \!\cdot\!
        \mathrm{erf}\left(\sqrt{ \frac{\hat{\rho}_m}{2\epsilon_m^1(\bm q_{-m})}}\right),
\end{align}
where $ \mathrm{erf}(\cdot)$ denotes the \emph{error function}~\cite{McDonough_SIGNAL}, and
\begin{align}
\label{equ: reform judging  threshold}
\hat{\rho}_m \! =\!\frac{\frac{1}{2}\sum_{\bm q_{-m}\in\mathcal Q_{-m}} \ln({\epsilon_m^1(\bm q_{-m})}/{\epsilon_m^0(\bm q_{-m})}) + \rho_m}%
    {\sum_{\bm q_{-m}\in\mathcal Q_{-m}} \!\frac{\epsilon_m^1(\bm q_{-m})  -  \epsilon_m^0(\bm q_{-m})}{\epsilon_m^1(\bm q_{-m})\cdot\epsilon_m^0(\bm q_{-m})}}.
\end{align}
\end{proposition}

\begin{proof}
Based on~(\ref{equ: judging function original}), the judging condition $\tau_m \LRT{\mathcal H_1}{\mathcal H_0} \rho_m$ is equivalent to $\mu_m^2 \LRT{\mathcal H_1}{\mathcal H_0} \hat{\rho}_m$.
Therefore, $\Pr\{\mu_m^2 > \hat{\rho}_m | \mathcal H_0\} = \Pr\{ \tau_m > \rho_m | \mathcal H_0\}$ 
and $\Pr\{\mu_m^2 \leq \hat{\rho}_m | \mathcal H_1\} = \Pr\{ \tau_m \leq \rho_m |\mathcal H_1 \}$.
Also, given $\bm q_{-m}$, $\mu_m^2$ follows a chi-squared distribution with one degree of freedom.
Therefore, the cumulative distribution function of $\mu_m^2$ is a weighted sum of error functions, and thus the conditional probabilities can be calculated by using~(\ref{equ: condi prob prop 1}) and~(\ref{equ: condi prob prop 2}).
\end{proof}

Besides, we can observe in~(\ref{equ: analytical CE loss of m}) that $L_m$ is determined by the judgment threshold $\rho_m$.
Then, based on~(\ref{equ: analytical CE loss of m}) to~(\ref{equ: reform judging  threshold}), $\partial L_m/\partial \rho_m$ can be calculated as
\beq
\label{equ: partial dev}
\partial L_m/\partial \rho_m = 
    -\frac{2C_{\mathrm{In0}}}{\sqrt{\pi}}
        \cdot
    \frac{\partial \hat{\rho}_m}{\partial \rho_m}
        \cdot
    \sum_{\bm q_{-m}\in\mathcal Q_{-m}}
            P_{m}(\bm q_{-m}) 
                \cdot
			\phi_m(\bm q_{-m}),
\eeq
\begin{align}
\phi_m(\bm q_{-m}) \! = \!
                \frac{(1\!-\! \tilde{p}_{m})\!\cdot \!e^{-\hat{\rho}_{m}/2\epsilon_m^0(\bm q_{-m}) }}{\sqrt{8\epsilon_m^0(\bm q_{-m})\hat{\rho}_m}}
                \! -\!
                \frac{\tilde{p}_{m}\cdot e^{-\hat{\rho}_{m}/2\epsilon_m^1(\bm q_{-m})} }{\sqrt{8\epsilon_m^1(\bm q_{-m})\hat{\rho}_m}}.
\end{align}
Then, the optimal $\rho_m^*$ can be obtained by solving $\partial L_m/\partial \rho_m = 0$.
Denoting the minimal $L_m$ corresponding to $\rho_m^*$ as $L_m^*$, the upper bound for the cross-entropy loss in~(\ref{equ: def of sensing loss}) can be calculated~as
\beq
L_{\mathrm{ub}} = \sum_{m\in\mathcal M} L_{m}^*.
\eeq

When the emptiness of the space grids other than the $m$-th is given, the upper bound of the cross-entropy loss can be calculated from Proposition~\ref{prop: opt judging threshold single q}.
Since the sensing accuracy is inversely proportional to the cross-entropy loss, a lower-bound for the sensing accuracy is derived. 

\begin{proposition}
\label{prop: opt judging threshold single q}
When the emptiness of the space grids other than the $m$-th is given, i.e., $\mathcal Q_{-m} = \{ \bm q_{-m} \}$%
, the optimal judging threshold for the $m$-th space grid is
\begin{align}
\label{equ: zero point}
\rho_m^*(\bm q_{-m}) \!=\! 
\begin{cases}
    \frac{1}{2}\ln(\frac{\epsilon_m^0(\bm q_{-m})}{\epsilon_m^1(\bm q_{-m})}), ~\text{if $\tilde{p}_{m}\!>\!{\sqrt{\epsilon_m^1(\bm q_{-m})\over \epsilon_m^0(\bm q_{-m})+\epsilon_m^1(\bm q_{-m})}}$}, \\
    2\ln(\frac{1 - \tilde{p}_{m}}{\tilde{p}_{m}}) - \frac{1}{2}\ln(\frac{\epsilon_m^0(\bm q_{-m})}{\epsilon_m^1(\bm q_{-m})}), ~\text{otherwise}.
\end{cases}
\end{align}
\end{proposition}
\begin{proof}
The sign of $\partial L_m/\partial p_m$ is determined by  $\phi_m(\bm q_{-m})$.
We calculate the ratio between the two terms of $\phi_m(\bm q_{-m})$, which can be expressed as
\beq
\iota_m(\bm q_{-m}) = 
    \frac{1-\tilde{p}_{m}}{\tilde{p}_{m}}
        \cdot
    \sqrt{\frac{\epsilon_m^1(\bm q_{-m})}{\epsilon_m^0(\bm q_{-m})}}
        \cdot
    e^{ -\hat{\rho}_m 
            \cdot 
        \frac{ \epsilon^1_{m}(\bm q_{-m}) - \epsilon^0_{m}(\bm q_{-m})  }{ 2\epsilon_m^0(\bm q_{-m})\cdot \epsilon_m^1(\bm q_{-m}) }
    }.
\eeq
Since $\epsilon_m^1(\bm q_{-m})>\epsilon_m^0(\bm q_{-m})$ and $\hat{\rho}_m\propto\rho_m$, $\iota_m(\bm q_{-m})$ is a monotonic decreasing function with respect to $\rho_m$ and $\iota_m(\bm q_{-m})\geq 0$.
Also,
$\phi_m(\bm q_{-m}) \geq 0 \iff \iota_m(\bm q_{-m}) \geq 1$, and thus, $\partial L_m/\partial p_m \geq 0$ if and only if $\iota_m(\bm q_{-m})\geq 1$.
Therefore, the minimal $H_m$ is obtained when $\rho_m$ satisfies the condition $\iota_m(\bm q_{-m}) = 1$.
Then, we can prove Proposition~\ref{prop: opt judging threshold single q} by solving $\iota_m(\bm q_{-m}) = 1$ and considering that $\rho_m \geq -\frac{1}{2}\ln(\frac{\epsilon_m^1(\bm q_{-m})}{\epsilon_m^0(\bm q_{-m})})$.
\end{proof}

However, since the number of possible $\bm q_{-m}$ can be large,~(typically, $|\mathcal Q_{-m}| = 2^{\bm M-1}$), calculating the exact $\partial L_m/\partial \rho_m$ in~(\ref{equ: partial dev}) is time-consuming, which makes it hard to find the exact $\rho_m^*$ and $L_m^*$.
Therefore, in practice, we approximate $H_{\mathrm{ub}}$ by using a random sampled subset of $\mathcal Q_{-m}^{\mathrm{sam}}$, which is denoted by $\mathcal Q_{-m}^{\mathrm{sam}} \subset \mathcal Q_{-m}$.

Moreover, since the sign of $\partial L_m/\partial \rho_m$ is determined by the sum of $\phi_m(\bm q_{-m})$, and $\phi_m(\bm q_{-m})$ has a zero point, which can be calculated by~(\ref{equ: zero point}).
If $\rho_m$ is less than the zero point of $\phi_m(\bm q_{-m})$, $\phi_m(\bm q_{-m})\geq 0$; and otherwise $\phi_m(\bm q_{-m}) < 0$.
Therefore, we use the mean of the optimal $\rho_m^*(\bm q_{-m})$ for each $\bm q_{-m} \in \mathcal Q_{-m}^{\mathrm{sam}}$ to estimate $\rho_m^*$, and approximate the upper-bound accordingly.
The estimated $\rho_m^*$ is denoted by $\tilde{\rho}_m^*$, which can be formulated as follows
\beq
\label{equ: final derived upper-bound}
\tilde{\rho}_m^* = 
	\frac{1}{|\mathcal Q_{-m}^{\mathrm{sam}}|} 
	\sum_{\bm q_{-m}\in\mathcal Q_{-m}^{\mathrm{sam}}} 
		\rho_m^*(\bm q_{-m}),
\eeq
where $\rho_m^*(\bm q_{-m})$ can be obtained by Proposition~\ref{prop: opt judging threshold single q}.
When $|\mathcal Q_{-m}^{\mathrm{sam}}|$ is large enough, $\tilde{\rho}_m^*$ in~(\ref{equ: final derived upper-bound}) can approximate $\rho_m^*$.

Finally, given the approximated upper bound of the cross-entropy loss as $\tilde{L}_{\mathrm{ub}}$, then it can be observed from~(\ref{equ: analytical CE loss of m}) that the upper bound of average probability of sensing error for a space grid is ${P}_{\mathrm{err, ub}} = \tilde{L}_{\mathrm{ub}} / C_{\mathrm{In0}}$.
Therefore, the lower bound of the average sensing accuracy for a space grid is ${P}_{\mathrm{acc, lb}} =1 - {P}_{\mathrm{err, ub}} $.

\section{Simulation and Evaluation}
\label{sec: simulation result}

In this section, we first describe the setting of the simulation scenario and summarize the simulation parameters.
Then, we provide simulation results to verify the effectiveness of the proposed PRPG algorithm.
Finally, using the proposed algorithm, we evaluate the cross-entropy loss of the metasurface assisted RF sensing scenario with respect to different numbers of sizes of the metasurface, and numbers of space grids.
\revised{Besides, we also compare the proposed method with the benchmark, i.e., the MIMO RF sensing systems.}

\subsection{Simulation Settings}
\label{ssec: simulation setting}
\begin{figure}[!t] 
	\center{\includegraphics[width=0.5\linewidth]{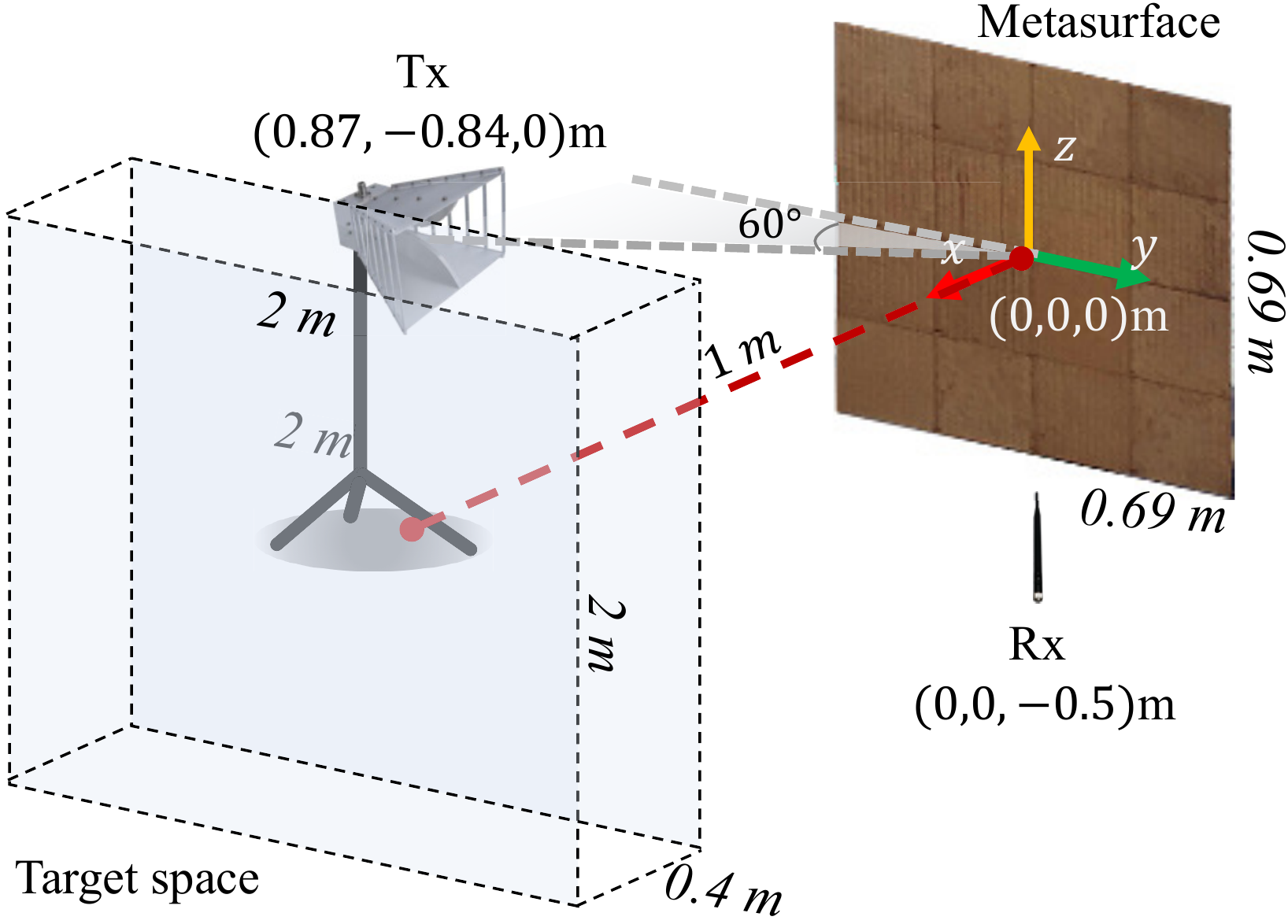}}
	\caption{Simulation layout.}
	\setlength{\belowcaptionskip}{-2em} 
		\vspace{-1.em}
	\label{fig: environment setup}
\end{figure}

The layout of the considered scenario is provided in Fig.~\ref{fig: environment setup}.
The metasurface adopted in this paper is the same as the one used in~\cite{our_ris_work}, and the reflection coefficients of the {\mselem} in different {\newstate}s are simulated in CST software, Microwave Studio, Transient Simulation Package~\cite{Hirtenfelder2007Effective}, by assuming $60^{\circ}$ incident RF signals with vertical polarization.
Besides, to increase the reflected signal power in the simulation, we combine $G$ {\mselem}s as an independently controllable group.
The {\mselem}s of an independently controllable group are in the same {\newstate}, and thus they can be considered as a single one.
Therefore, the proposed algorithm is suitable for this case. 
The number of independently controllable group is denoted by $N_G$.

The origin of the coordinate is at the center of the metasurface, and the metasurface is in the $y$-$z$ plane.
In addition, the $z$-axis is vertical to the ground and pointing upwards, and the $x$- and $y$-axes are parallel to the ground.
The Tx and Rx antennas are located at $(0.87, -0.84, 0)$~m and $(0,0,-0.5)$~m, respectively.
The target space a cuboid region located at $1$~m from the metasurface, and is divided into $M$ space blocks each with size $0.1\times 0.1\times 0.1$~m$^3$.
The simulation parameters are summarized in Table~\ref{table: simul and experi parameters}.

\begin{table}
\centering
	\caption{Simulation Parameters.}	\label{table: simulation parameters}
	\vspace{-1em}
\includegraphics[width=0.5\linewidth]{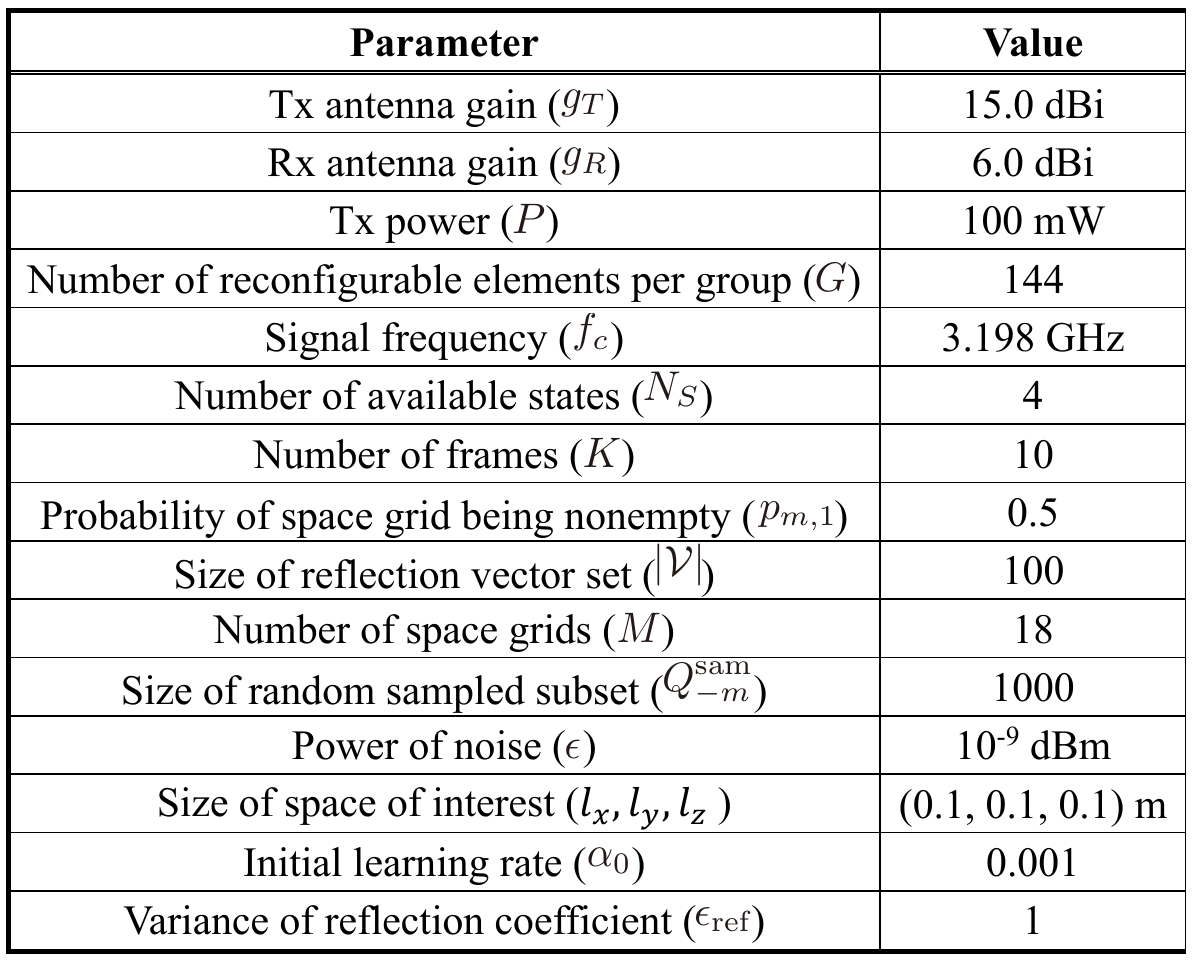}
\label{table: simul and experi parameters}
\end{table}

\begin{figure}[!t] 
	\center{\includegraphics[width=0.5\linewidth]{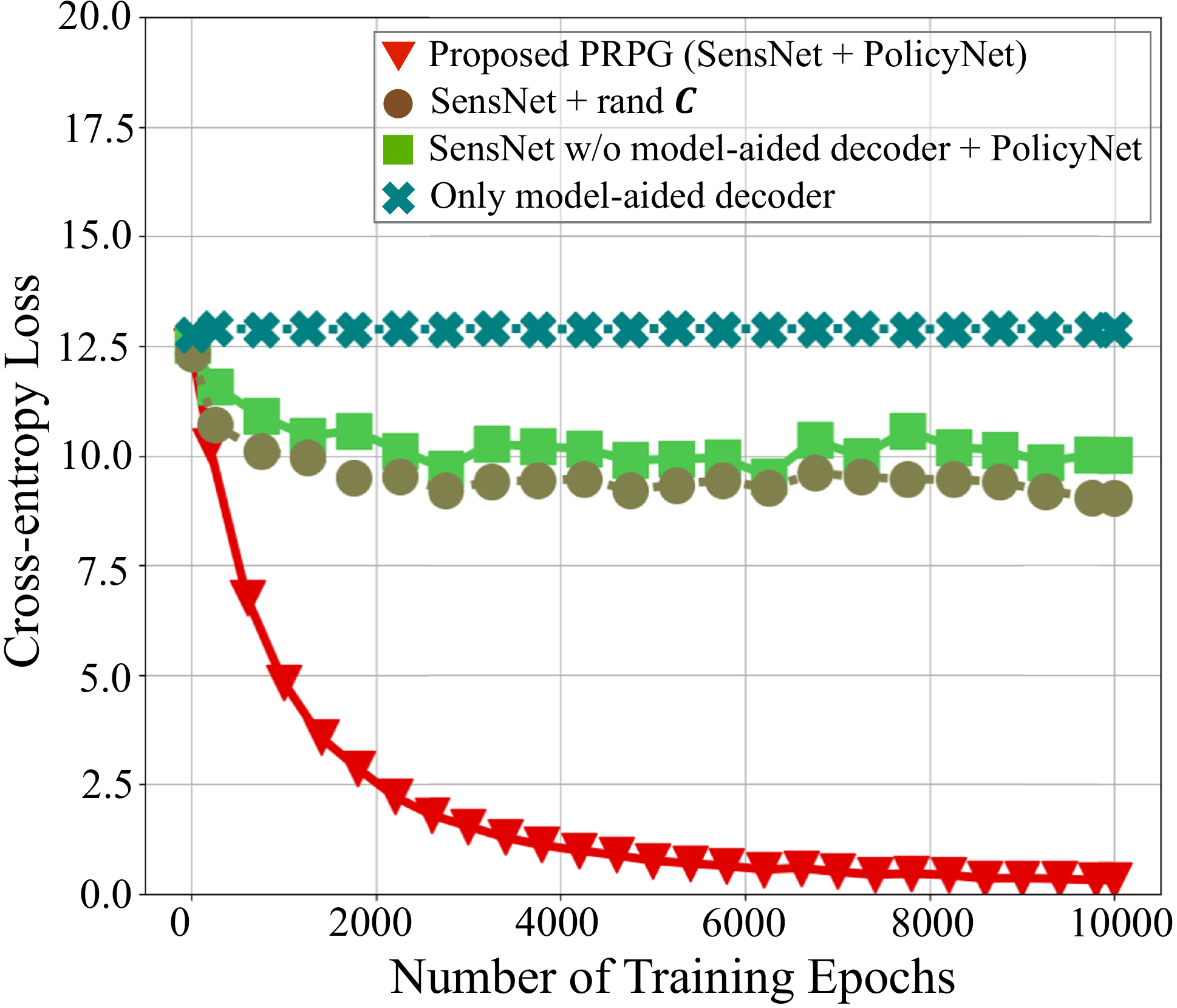}}
	\vspace{-0.5em}
	\setlength{\belowcaptionskip}{-1.2em}   
	\caption{Cross-entropy loss versus the number of training epochs for different algorithms.}
	\label{fig: algorithm compare cross-entropy loss}
\end{figure}

\subsection{Results}
\label{ssec: result}

In Fig.~\ref{fig: algorithm compare cross-entropy loss}, we compare the training results for different algorithms.
Specifically, the first algorithm in the legend is the proposed PRPG algorithm where a sensing network~(SensNet) and a policy network~(PolicyNet) are adopted.
The second algorithm adopts a sensing network but adopt a random {\configmat}.
The third algorithm adopts both a sensing network and a policy network, but the sensing network does not contain a model-aided decoder as in the proposed algorithm.
The fourth algorithm only uses the model-aided decoder to map the received signals to the sensing results.

It can be observed that the proposed PRPG algorithm converges with high speed and it results in the lowest cross-entropy loss among all the considered algorithms.
In particular, Fig.~\ref{fig: sensing result versus training episodes} shows a ground-truth object and the corresponding sensing results versus the number of training epochs.
As the number of training epochs increases, the sensing result approaches the ground truth and becomes approximately the same after $10^4$ training epochs.

\begin{figure}[!t] 
	\center{\includegraphics[width=1\linewidth]{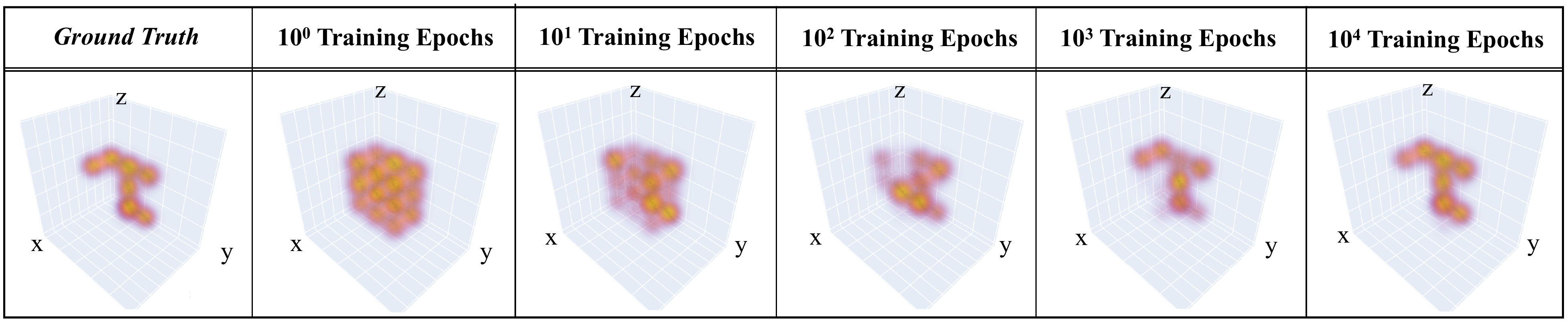}}
	\vspace{-0.6em}
	\setlength{\belowcaptionskip}{-1.2em}   
	\caption{Illustrations of ground-truth and the sensing results of different training epochs for a target object.}
	\vspace{-1em}
	\label{fig: sensing result versus training episodes}
\end{figure}

\begin{figure}[!t] 
	\center{\includegraphics[width=0.95\linewidth]{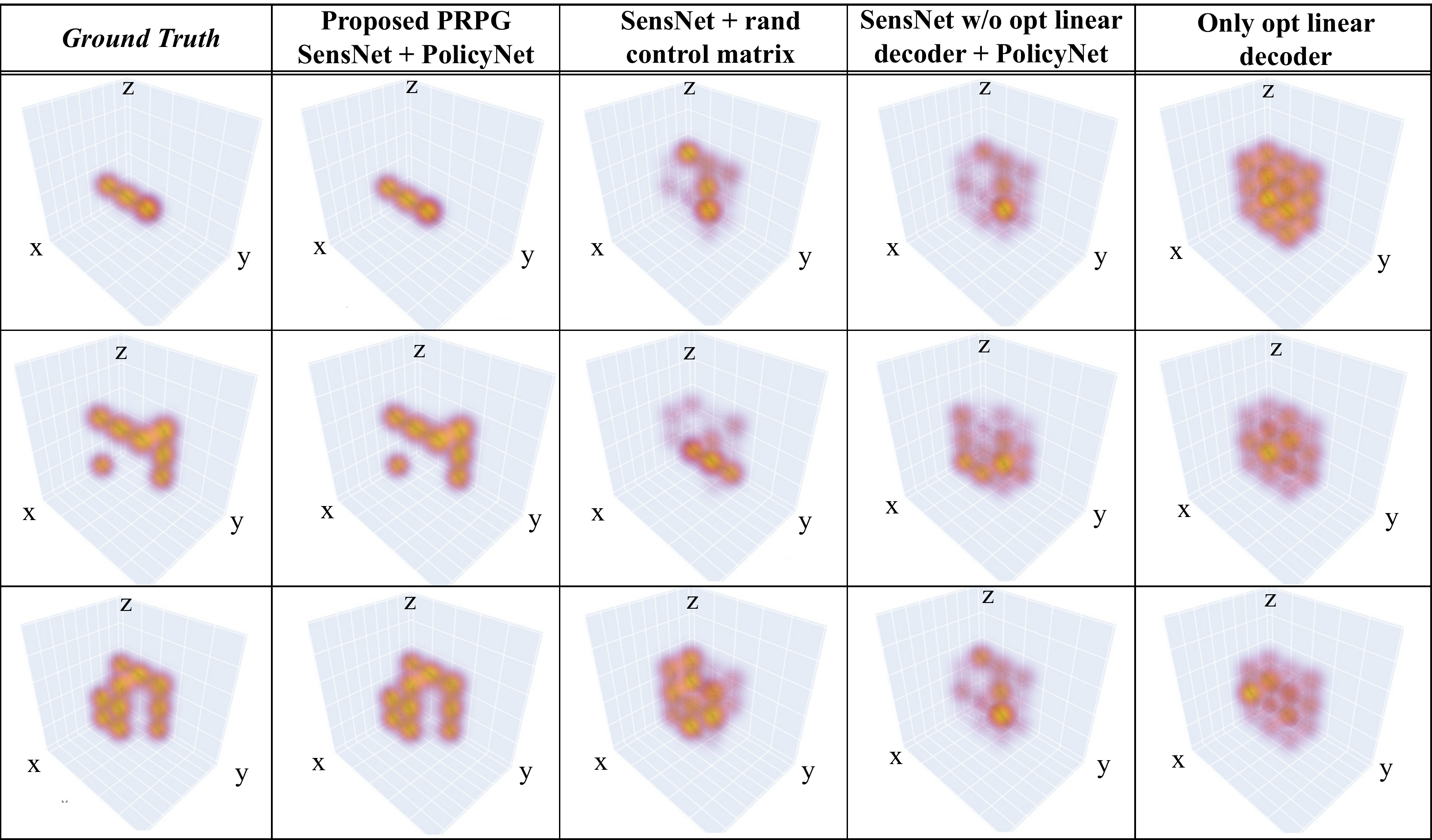}}
	\vspace{-0.5em}
	\setlength{\belowcaptionskip}{-1.2em}   
	\caption{\revised{Illustrations of the ground-truths and the sensing results of objects with different shapes for different algorithms.}}
	\label{fig: image result comparing for different algs and shapes.}
\end{figure}

\begin{figure}[!t] 
	\center{\includegraphics[width=0.5\linewidth]{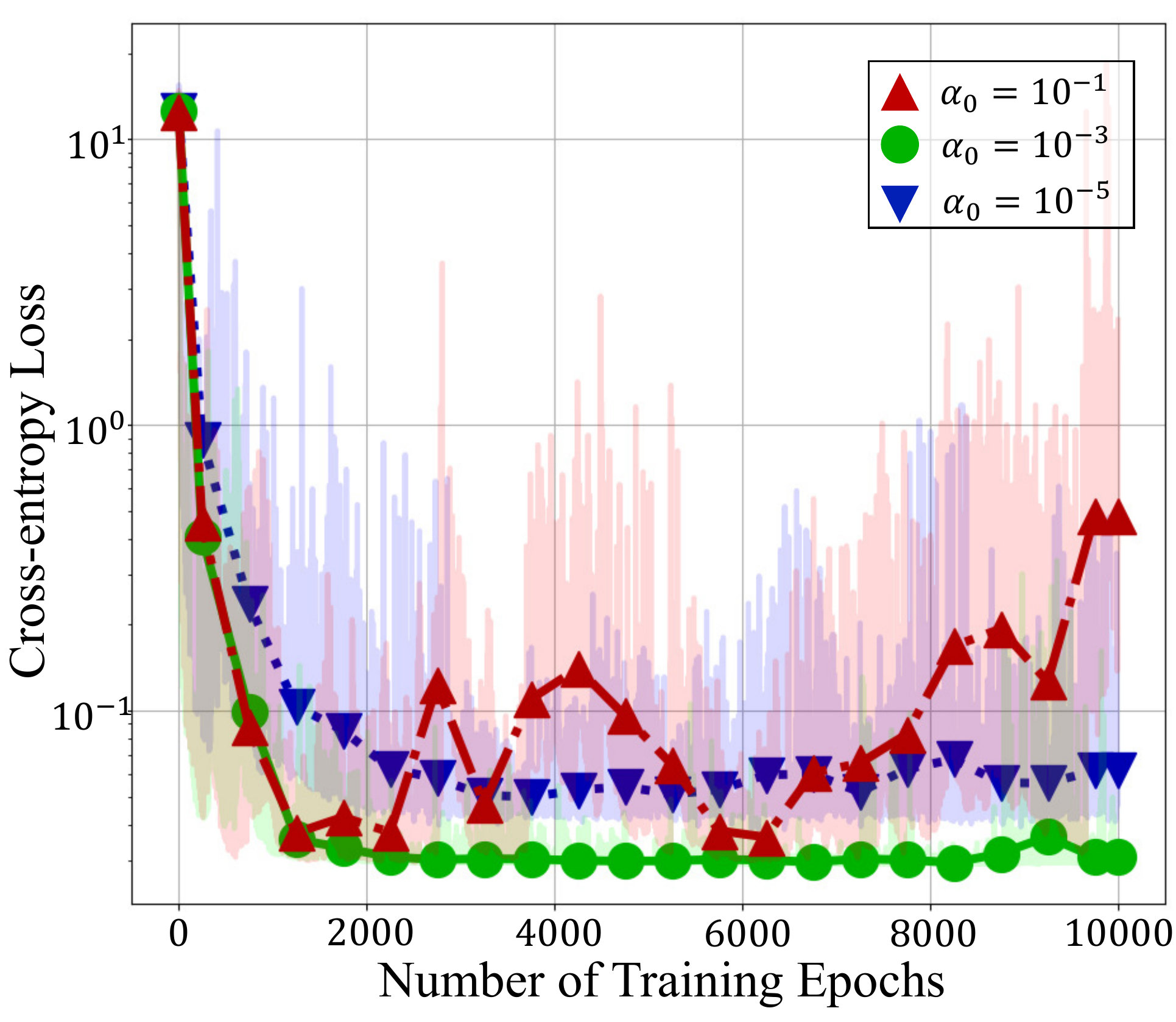}}
	\vspace{-0.5em}
	\setlength{\belowcaptionskip}{-1.2em}   
	\caption{Cross-entropy loss of the {\sensfunc} in high, normal, and low learning rate cases.}
 \vspace{-1em}
	\label{fig: lr compare}
\end{figure}

\revised{In Fig.~\ref{fig: image result comparing for different algs and shapes.}, it shows the ground-truths and the sensing results for different algorithms and the target objects with different shapes.
Comparing the sensing results with the ground truths, we can observe that the proposed algorithm outperforms other benchmark algorithms to a large extent.
Besides, by comparing the sensing results of the proposed algorithm in the second column with the ground truths in the first column, we can observe that the proposed algorithm obtains the accurate sensing results despite the different shapes of the target objects.}''

In Fig.~\ref{fig: lr compare}, the training results versus the number of training epochs for the PRPG algorithm are given and compared for different learning rates are compared.
The initial learning rates in each case are set to be $\alpha_0 = 10^{-1}$, $10^{-3}$, $10^{-5}$, which then decrease inversely as the number of training epochs increases.
It can be observed that large values of the learning rates prevent the algorithm to converge, while low values of the learning rates result in a slow decrease of the cross entropy loss.
The setup $\alpha_0=10^{-3}$ outperforms the others, which verifies our learning rate selection.

\begin{figure}[!t] 
	\center{\includegraphics[width=0.58\linewidth]{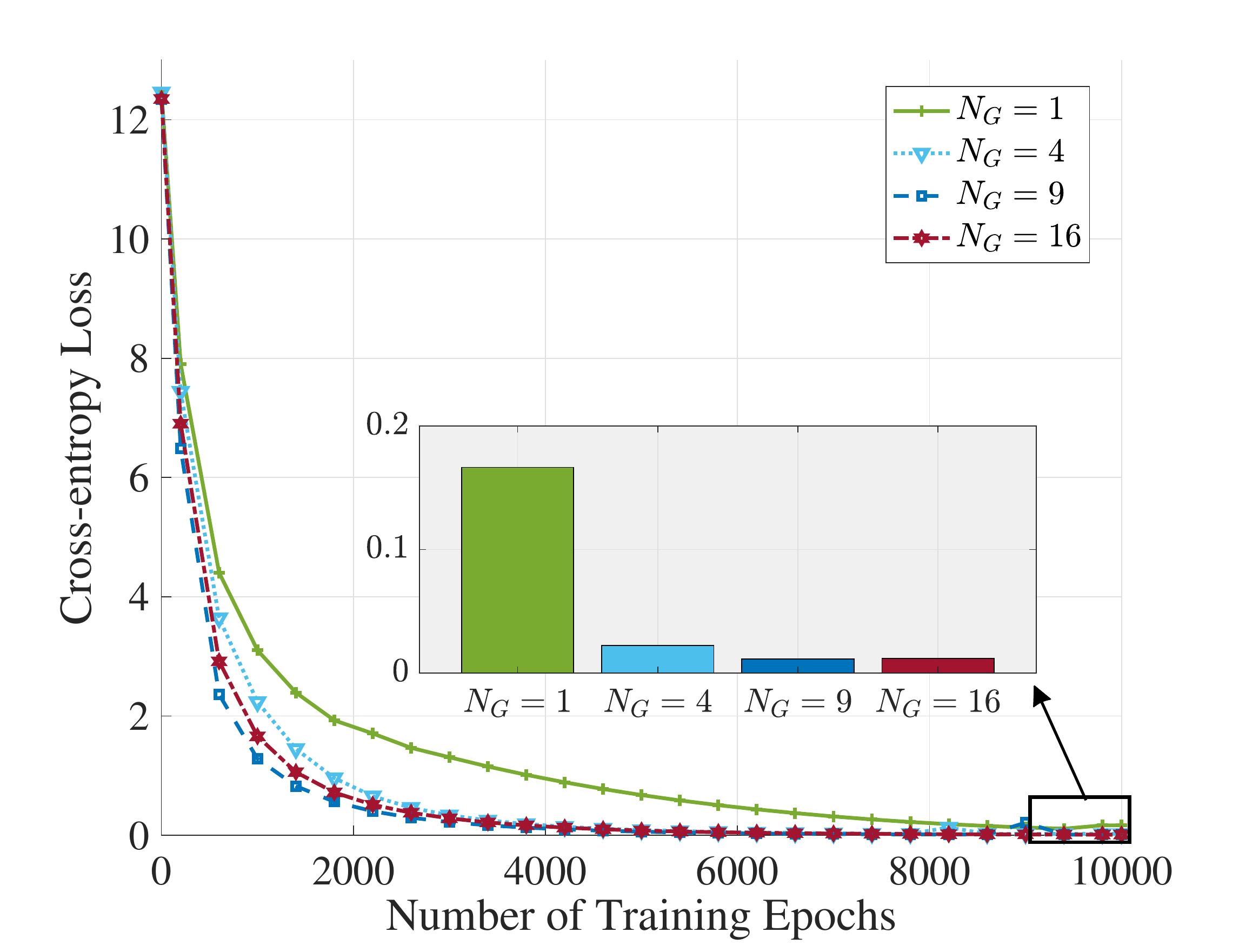}}
	\vspace{-0.5   em}
	\setlength{\belowcaptionskip}{-1.2em}   
	\caption{\revised{Cross-entropy loss versus the number of training epochs for different sizes of the metasurface.}}
\vspace{-1em}
	\label{fig: size compare}
\end{figure}

\revised{In Fig.~\ref{fig: size compare}, it can be observed that as the size of the metasurface, i.e., $N_G$, increases, the result cross-entropy loss after training decreases.
This is because the received energy can be improved with more reconfigurable elements to reflect transmitted signals, as indicated by~(\ref{equ: received signal matrix form}).
Besides, more {\mselem}s create a larger design freedom and higher controllability of the beamforming, which makes gains of these reflection paths via different space grids more distinguishable.
Therefore, objects at different space grids can be sensed with a higher precision.
However, the cross-entropy cannot be reduced infinitely. 
When $N_G$ is sufficiently large, the cross-entropy will remains stable.
As shown in Fig.~\ref{fig: size compare}, the cross-entropy loss results for $N_G=9$ and $N_G=16$ are almost the same.
Besides, comparing the curves for $N_G=9$ and $N_G=16$ within the first $2000$ training epochs, we can observe that increasing the number of {\mselem}s when $N_G\geq 9$ has a negative impact on the training speed and convergence rate.
This is because increasing the number of {\mselem}s leads to a higher complexity of finding the optimal policy for the metasurface to determine its {\configmat}, since the policy network of the metasurface needs to handle a higher-dimensional state space.}

\begin{figure}[!t] 
	\center{\includegraphics[width=0.5\linewidth]{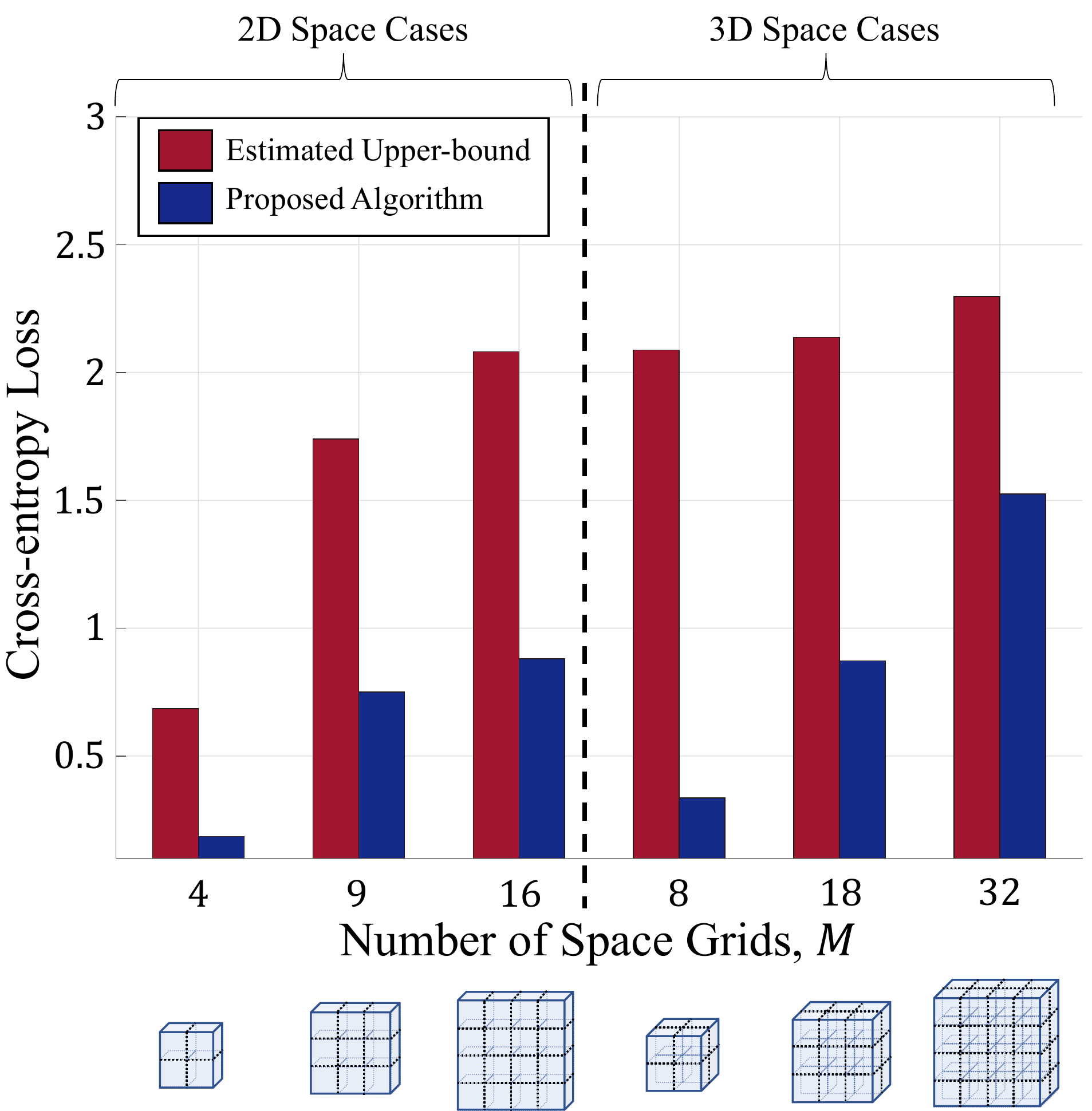}}
	\vspace{-0.5   em}
	\setlength{\belowcaptionskip}{-1.2em}   
	\caption{Estimated upper-bound and the results of the proposed algorithm for the cross-entropy loss versus different numbers of space grids in $2$D and $3$D scenarios. The drawings at the bottom indicate the arrangement of the space grids}
\vspace{-1em}
	\label{fig: theo and algorithm result compare}
\end{figure}

In Fig.~\ref{fig: theo and algorithm result compare}, we compare the theoretical upper-bound derived in~(\ref{equ: final derived upper-bound}) and the proposed PRPG algorithm for different values of $M$ in $2$D and $3$D scenarios.
It can be observed that, in both $2$D and $3$D scenarios, the probability of sensing error increases with $M$.
Also, the cross-entropy loss in $3$D scenarios is higher than those for $2$D scenarios.
This is because the space grids in the $3$D scenarios are more closely spaced to each other, which make them hard to be distinguished.
Finally, it can be observed that, as $M$ increases, the cross-entropy loss of the proposed algorithm increases more quickly in 3D scenarios compared to that in $2$D scenarios.
This which verifies that $3$D sensing is more difficult than $2$D sensing.

\begin{figure}[!t] 
	\center{\includegraphics[width=0.55\linewidth]{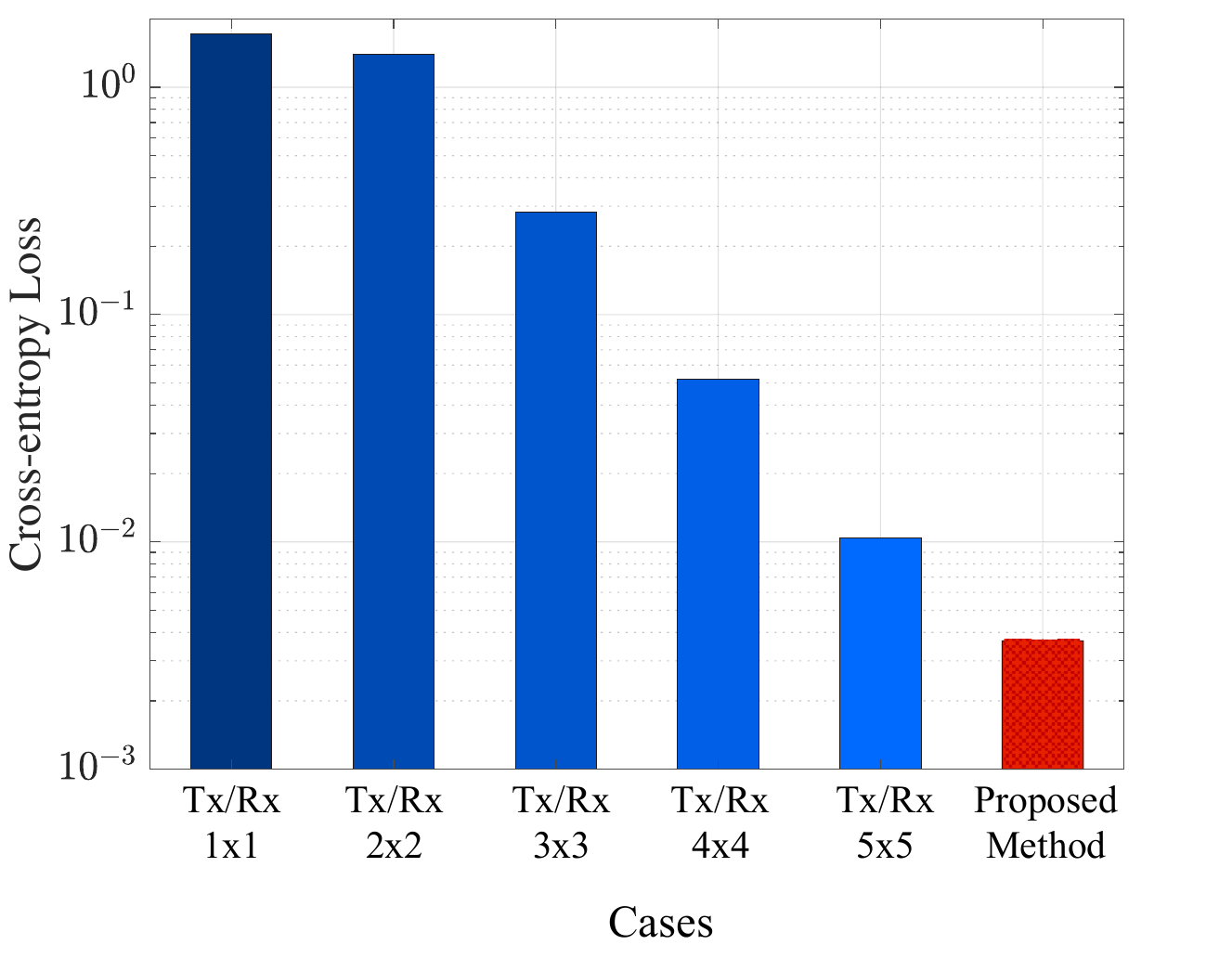}}
	\vspace{-0.5   em}
	\setlength{\belowcaptionskip}{-1.2em}   
	\caption{\revised{Comparison between the metasurface assisted scenario and the MIMO scenarios with different numbers of Tx/Rx antennas. The bars illustrate the results for different MIMO scenarios, and the dash lines depict the results of the metasurface assisted scenario with different numbers of frames.}}
\vspace{-1em}
	\label{fig: mimo compare}
\end{figure}

\revised{In Fig.~\ref{fig: mimo compare}, we show the comparison between the proposed metasurface assisted scenario and the benchmark, which is the MIMO RF sensing scenarios with no metasurface.
Both the metasurface assisted scenario and the MIMO scenarios adopted a similar layout described in Section~\ref{ssec: simulation setting}, and the result cross-entropy loss is obtained by Algorithm~\ref{alg: proposed algorithm}.
Nevertheless, in the MIMO sensing scenarios, a static reflection surface takes the place of the metasurface, which cannot change the {\config} for the reflection signals.
When the size of the MIMO array in Fig.~\ref{fig: mimo compare} is $n\times n$~($n=1,2,\dots,5$), it indicates that $n$ Tx antennas and $n$ Rx antennas are adopted in the scenario.
Specifically, the Tx/Rx antennas are arranged along the $y$-axis with a space interval of $0.1$ m.
The Tx antennas transmit continuous signals with phase interval $2\pi/n$, and the $n$ received signals of the Rx antennas with suppressed LoS signals are used as the measurement vector.}
\revised{Comparing the MIMO benchmarks and the proposed method, we can observe that the proposed method outperforms the $(1\times 1)\sim(5\times 5)$ MIMO sensing scenarios in terms of the resulted cross-entropy loss. This shows the significance of using metasurface to assisted RF sensing.}

%

\section{Conclusion}
\label{sec: conclu}
In this paper, we have considered a metasurface assisted RF sensing scenario, where the existence and locations of objects within a $3$D target space can be sensed.
To facilitate the {\config} design, we have proposed a frame-based RF sensing protocol.
Based on the proposed protocol, we have formulated an optimization problem to design of the {\config} and the {\sensfunc}.
To solve the optimization problem, we have formulated an MDP framework and have proposed a deep reinforcement learning algorithm named PRPG to solve it.
Also, we have analyzed the computational complexity and the convergence of the proposed algorithm and have computed a theoretical upper-bound for the cross-entropy loss given the {\config}s of the metasurface.
Simulation results have verified that the proposed algorithm outperforms other benchmark algorithms in terms of training speed and the resulted cross-entropy loss. 
Besides, they have also illustrated the influence of the sizes of the metasurface and the target space on the cross-entropy loss, which provides insights on the implementation of practical metasurface assisted RF sensing systems.

\begin{appendices}
\section{Proof of Theorem \ref{theo: action determine complexity}}
\label{appx: theorem 1}

The complexities of embedding $K$ and $N$ as one-hot vectors is $\bigO(K)$ and $\bigO(N)$, respectively.
Based on~\cite{skiena2012sorting}, the complexity of multiplex the $K$ {\config}s that are $N\cdot N_S$-dimensional by $\bm A\in\mathbb C^{NN_S\times M}$ is $\bigO(K\cdot N\cdot N_S\cdot M)$.
For a fully connected neural network with a fixed number of hidden layers and neurons, the computational complexity of the back-propagation algorithm is proportional to the product of the input size and the output size~\cite{sipper1993serial}.
Therefore, the computational complexities of the symmetric MLP group and the Q-value MLP are $\bigO(K\cdot M)$ and $\bigO(K\cdot N_S+N\cdot N_S)$, respectively.
As the connecting operation has complexity $\bigO(1)$ and finding the maximum Q-value is $\bigO(N_S)$, the total computational complexity has complexity $\bigO(KNN_SM)$ and is therefore dominated by the matrix multiplication.\proofend

\section{Proof of Lemma \ref{lemma: reward calculation}}
\label{appx: proof of lemma 1}
We consider the worst case scenario for the computation, i.e., the former states in all the samples are terminal states.
In this case, the rewards are calculated from~(\ref{equ: calculate r by monte carlo}). 
The term inside the second summation consists of two part, i.e., the cross-entropy calculation which has computational complexity $\bigO(M)$, and the calculation of $\hat{\bm p}$ by using the sensing network.
The computational complexity of calculating $\bm C \bm A$ is $\bigO(KNN_SM)$.

Based on~\cite{PETKOVIC2009270}, calculating the pseudo-inverse matrix $\nGamma^+$, where $\nGamma$ is a $K\times M$ matrix, has complexity $\bigO(K^2M)$.
Similar to the analysis of the computational complexity of MLPs in the proof of Theorem~\ref{theo: action determine complexity}, the computational complexity of the MLP is $\bigO(M^2)$.
Therefore, the computational complexity of calculating $\hat{\bm p}$ is $\bigO(KNN_SM+K^2M+M^2)$, which proves Lemma~\ref{lemma: reward calculation}.\proofend

\section{Proof of Lemma \ref{lemma: training calculation complexity}}
\label{appx: proof of lemma 2}
For a fully connected neural network with a fixed number of hidden layers and neurons, the computational complexity of the back-propagation algorithm is proportional to the product of the input size and the output size~\cite{sipper1993serial}.
Therefore, the computational complexity of using the back-propagation algorithm for updating the parameter vector of the sensing network is $\bigO(M^2)$.

The policy network can be considered as two connected MLPs: the first one takes the one-hot embedding vectors of $k$ and $n$ as the input, and the second one takes the $K$ measuring vectors with $2M$ dimensions as the input.
Moreover, as a symmetric MLP group is considered, the actual size of the input vector for the second MLP is $2M$ instead of $2KM$.
Therefore, the computational complexity of training the first and second MLP of the policy network are $\bigO(N_S\cdot(K+N))$ and $\bigO(N_SM)$, respectively, and the total computational complexity is thus $\bigO(N_S\cdot(K+N+M))$.

Furthermore, if a single large MLP with layer sizes $(2KM, 64, 32K)$ is used to substitute the symmetric MLP group, the computational complexities of training the second MLP is $\bigO(KN_SM)$, and the total computational complexity of training the policy network is $\bigO(KMN_S+NN_S)$.
Therefore, Lemma~\ref{lemma: training calculation complexity} is proved. \proofend

\section{Proof of Theorem \ref{theo: training complexity}}
\label{appx: proof of theorem 2}

Based on~(\ref{equ: update of theta}) and~(\ref{equ: sensing function loss}), the complexity of calculating the loss functions are determined by the computation of the reward and action probabilities.
From Theorem~\ref{theo: action determine complexity} and Lemma~\ref{lemma: training calculation complexity}, it follows that the complexity of calculating the reward is of higher order than that of calculating the action probabilities.
Therefore, the computational complexity of the training phase is dominated by the calculation of the $N_b$ rewards, which is $\bigO\!\left(KNN_SM+K^2M+M^2 \right)$. 
\proofend

\end{appendices}

\begin{footnotesize}
\bibliographystyle{IEEEtran}
\bibliography{ms}

\begin{thebibliography}{10}
\providecommand{\url}[1]{#1}
\csname url@samestyle\endcsname
\providecommand{\newblock}{\relax}
\providecommand{\bibinfo}[2]{#2}
\providecommand{\BIBentrySTDinterwordspacing}{\spaceskip=0pt\relax}
\providecommand{\BIBentryALTinterwordstretchfactor}{4}
\providecommand{\BIBentryALTinterwordspacing}{\spaceskip=\fontdimen2\font plus
\BIBentryALTinterwordstretchfactor\fontdimen3\font minus
  \fontdimen4\font\relax}
\providecommand{\BIBforeignlanguage}[2]{{%
\expandafter\ifx\csname l@#1\endcsname\relax
\typeout{** WARNING: IEEEtran.bst: No hyphenation pattern has been}%
\typeout{** loaded for the language `#1'. Using the pattern for}%
\typeout{** the default language instead.}%
\else
\language=\csname l@#1\endcsname
\fi
#2}}
\providecommand{\BIBdecl}{\relax}
\BIBdecl

\bibitem{Kianoush2017Device}
S.~{Kianoush}, S.~{Savazzi}, F.~{Vicentini}, V.~{Rampa}, and M.~{Giussani},
  ``Device-free {RF} human body fall detection and localization in industrial
  workplaces,'' \emph{IEEE Internet of Things J.}, vol.~4, no.~2, pp. 351--362,
  Apr. 2017.

\bibitem{Lee2009Wireless}
P.~W.~Q. {Lee}, W.~K.~G. {Seah}, H.~{Tan}, and Z.~{Yao}, ``Wireless sensing
  without sensors -- {A}n experimental approach,'' in \emph{Proc. IEEE Int.
  Symp. Pers. Indoor Mobile Radio Commun.}, Tokyo, Japan, Sep. 2009.

\bibitem{He2006Vigilnet}
T.~He, S.~Krishnamurthy, L.~Luo, T.~Yan, L.~Gu, R.~Stoleru, G.~Zhou, Q.~Cao,
  P.~Vicaire, J.~A. Stankovic, T.~F. Abdeizaher, J.~Hui, and B.~Krogh,
  ``Vigilnet: An integrated sensor network system for energy-efficient
  surveillance,'' \emph{ACM Trans. Sensor Netw.}, vol.~2, no.~1, pp. 1--38,
  Feb. 2006.

\bibitem{Ota2018QUOIN}
K.~{Ota}, M.~{Dong}, J.~{Gui}, and A.~{Liu}, ``Quoin: Incentive mechanisms for
  crowd sensing networks,'' \emph{IEEE Netw.}, vol.~32, no.~2, pp. 114--119,
  Mar. 2018.

\bibitem{cook2009assessing}
D.~J. Cook and M.~Schmitter-Edgecombe, ``Assessing the quality of activities in
  a smart environment,'' \emph{Methods Inform. Medicine}, vol.~48, no.~5, pp.
  480--485, Oct. 2009.

\bibitem{amin2016radar}
M.~G. Amin, Y.~D. Zhang, F.~Ahmad, and K.~D. Ho, ``Radar signal processing for
  elderly fall detection: The future for in-home monitoring,'' \emph{IEEE
  Signal Process. Mag.}, vol.~33, no.~2, pp. 71--80, Mar. 2016.

\bibitem{zhang2019Feasibility}
D.~Zhang, J.~Wang, J.~Jang, J.~Zhang, and S.~Kumar, ``On the feasibility of
  {Wi-Fi} based material sensing,'' in \emph{Proc. MobiCom}, Los Cabos, Mexico,
  Oct. 2019.

\bibitem{jiang2018towards}
W.~Jiang, C.~Miao, F.~Ma, S.~Yao, Y.~Wang, Y.~Yuan, H.~Xue, C.~Song, X.~Ma,
  D.~Koutsonikolas, and et~al., ``Towards environment independent device free
  human activity recognition,'' in \emph{Proc. MobiCom}, New Delhi, India, Oct.
  2018.

\bibitem{Hsu2019Enabling}
C.-Y. Hsu, R.~Hristov, G.-H. Lee, M.~Zhao, and D.~Katabi, ``Enabling
  identification and behavioral sensing in homes using radio reflections,'' in
  \emph{Proc. CHI}, Glasgow, U.K., May 2019.

\bibitem{adib2015capturing}
F.~Adib, C.-Y. Hsu, H.~Mao, D.~Katabi, and F.~Durand, ``Capturing the human
  figure through a wall,'' \emph{ACM Trans. Graphics}, vol.~34, no.~6, p. 219,
  Oct. 2015.

\bibitem{zhao2018rf}
M.~Zhao, Y.~Tian, H.~Zhao, M.~A. Alsheikh, T.~Li, R.~Hristov, Z.~Kabelac,
  D.~Katabi, and A.~Torralba, ``Rf-based 3d skeletons,'' in \emph{Proc.
  SIGCOMM}, New York, NY, Aug. 2018.

\bibitem{PedrossEngel2018Orthogonal}
A.~Pedross-Engel, D.~Arnitz, J.~N. Gollub, O.~Yurduseven, K.~P. Trofatter,
  M.~F. Imani, T.~Sleasman, M.~Boyarsky, X.~Fu, D.~L. Marks, D.~R. Smith, and
  M.~S. Reynolds, ``Orthogonal coded active illumination for millimeter wave,
  massive-{MIMO} computational imaging with metasurface antennas,'' \emph{IEEE
  Trans. Comput. Imaging}, vol.~4, no.~2, pp. 184--193, Jun. 2018.

\bibitem{honma2018Human}
N.~{Honma}, D.~{Sasakawa}, N.~{Shiraki}, T.~{Nakayama}, and S.~{Iizuka},
  ``Human monitoring using {MIMO} radar,'' in \emph{Proc. iWEM}, Nagoya, Japan,
  Aug. 2018.

\bibitem{li2020programmable}
Z.~Li, Y.~Xie, L.~Shangguan, R.~I. Zelaya, J.~Gummeson, W.~Hu, and K.~Jamieson,
  ``Programmable radio environments with large arrays of inexpensive
  antennas,'' \emph{GetMobile: Mobile Comput. Commun.}, vol.~23, no.~3, pp.
  23--27, Sep. 2019.

\bibitem{Renzo2019Smart}
M.~Di~Renzo, M.~Debbah, D.-T. Phan-Huy, A.~Zappone, M.-S. Alouini, C.~Yuen,
  V.~Sciancalepore, G.~C. Alexandropoulos, J.~Hoydis, H.~Gacanin, J.~D. Rosny,
  A.~Bounceu, G.~Lerosey, and M.~Fink, ``Smart radio environments empowered by
  ai reconfigurable meta-surfaces: An idea whose time has come,''
  \emph{arXiv:1903.08925}.

\bibitem{Gacanin2020Wireless}
H.~Gacanin and M.~D. Renzo, ``Wireless 2.0: Towards an intelligent radio
  environment empowered by reconfigurable meta-surfaces and artificial
  intelligence,'' \emph{arXiv:2002.11040}.

\bibitem{Basar2019Generalization}
E.~Basar, M.~D. Renzo, J.~D. Rosny, M.~Debbah, M.~Alouini, and R.~Zhang,
  ``Wireless communications through reconfigurable intelligent surfaces,''
  \emph{IEEE Access}, vol.~7, pp. 116\,753--116\,773, Aug. 2019.

\bibitem{Zhang2020Reflective}
S.~Zhang, H.~Zhang, B.~Di, Y.~Tan, Z.~Han, and L.~Song,
  ``Reflective-transmissive metasurface aided communications for
  full-dimensional coverage extension,'' \emph{IEEE Trans. Veh. Technol.},
  early access.

\bibitem{ElMossallamy2020Reconfigurable}
M.~A. {ElMossallamy}, H.~{Zhang}, L.~{Song}, K.~G. {Seddik}, Z.~{Han}, and
  G.~Y. {Li}, ``Reconfigurable intelligent surfaces for wireless
  communications: Principles, challenges, and opportunities,'' \emph{IEEE
  Trans. Cognitive Commun. Netw.}, vol.~6, no.~3, pp. 990--1002, Sep. 2020.

\bibitem{Hashida2020Intelligent}
H.~{Hashida}, Y.~{Kawamoto}, and N.~{Kato}, ``Intelligent reflecting surface
  placement optimization in air-ground communication networks toward {6G},''
  \emph{IEEE Wireless Commu.}, early access.

\bibitem{Li2019Machine}
L.~Li, H.~Ruan, C.~Liu, Y.~Li, Y.~Shuang, A.~Al{\`u}, C.-W. Qiu, and T.~J. Cui,
  ``Machine-learning reprogrammable metasurface imager,'' \emph{Nature
  Commun.}, vol.~10, no. 1082, pp. 1--8, Jun. 2019.

\bibitem{Zhang2020Towards}
H.~Zhang, H.~Zhang, B.~Di, K.~Bian, Z.~Han, and L.~Song, ``Towards ubiquitous
  positioning by leveraging reconfigurable intelligent surface,'' \emph{IEEE
  Commun. Lett.}, early access.

\bibitem{sutton1998reinforcement}
R.~S. Sutton and A.~G. Barto, \emph{Reinforcement Learning: An
  Introduction}.\hskip 1em plus 0.5em minus 0.4em\relax Cambridge, MA: MIT
  Press, 2018.

\bibitem{Liu2019Reconfigurable}
Y.~Liu, X.~Liu, T.~Hou, J.~Xu, Z.~Qin, M.~D. Renzo, and N.~Al-Dhahir,
  ``Reconfigurable intelligent surfaces: Principles and opportunities,''
  \emph{arXiv:2007.03435}.

\bibitem{Dai2020Reconfigurable}
L.~{Dai}, B.~{Wang}, M.~{Wang}, X.~{Yang}, J.~{Tan}, S.~{Bi}, S.~{Xu},
  F.~{Yang}, Z.~{Chen}, M.~D. {Renzo}, C.~B. {Chae}, and L.~{Hanzo},
  ``Reconfigurable intelligent surface-based wireless communications: Antenna
  design, prototyping, and experimental results,'' \emph{IEEE Access}, vol.~8,
  pp. 45\,913--45\,923, Mar. 2020.

\bibitem{di2019hybrid}
B.~Di, H.~Zhang, L.~Li, L.~Song, Z.~Han, and H.~V. Poor, ``Hybrid beamforming
  for reconfigurable intelligent surface based multi-user communications:
  Achievable rates with limited discrete phase shifts,'' \emph{IEEE J. Sel.
  Areas Commun.}, vol.~38, no.~8, pp. 1809--1822, Aug. 2020.

\bibitem{our_ris_work}
J.~Hu, H.~Zhang, B.~Di, L.~Li, K.~Bian, L.~Song, Y.~Li, Z.~Han, and H.~V. Poor,
  ``Reconfigurable intelligent surface based {RF} sensing: Design,
  optimization, and implementation,'' \emph{IEEE J. Sel. Areas Commun.}, early
  access, \emph{arXiv:}1912.09198.

\bibitem{goldsmith2005wireless}
A.~Goldsmith, \emph{Wireless communications}.\hskip 1em plus 0.5em minus
  0.4em\relax Cambridge University Press, 2005.

\bibitem{tang2019wireless}
W.~Tang, M.~Z. Chen, X.~Chen, J.~Y. Dai, Y.~Han, M.~Di~Renzo, Y.~Zeng, S.~Jin,
  Q.~Cheng, and T.~J. Cui, ``Wireless communications with reconfigurable
  intelligent surface: Path loss modeling and experimental measurement,''
  \emph{arXiv:1911.05326}.

\bibitem{zhang2019reconfigurable}
H.~Zhang, B.~Di, L.~Song, and Z.~Han, ``Reconfigurable intelligent surfaces
  assisted communications with limited phase shifts: {H}ow many phase shifts
  are enough?'' \emph{IEEE Trans. Veh. Technol.}, vol.~69, no.~4, pp.
  4498--4502, Feb. 2020.

\bibitem{goodfellow2016deep}
I.~Goodfellow, Y.~Bengio, and A.~Courville, \emph{Deep Learning}.\hskip 1em
  plus 0.5em minus 0.4em\relax Cambridge, MA: MIT press, 2016.

\bibitem{Boyd_CONVEX}
S.~Boyd and L.~Vandenberghe, \emph{Convex optimization}.\hskip 1em plus 0.5em
  minus 0.4em\relax Cambridge University Press: Cambridge, U.K., 2004.

\bibitem{McDonough_SIGNAL}
R.~N. McDonough and A.~D. Whalen, \emph{Detection of signals in noise}.\hskip
  1em plus 0.5em minus 0.4em\relax Academic Press: San Diego, CA, USA, 2004.

\bibitem{bezdek2003convergence}
J.~C. Bezdek and R.~J. Hathaway, ``Convergence of alternating optimization,''
  \emph{Neural, Parallel \& Sci. Comput.}, vol.~11, no.~4, pp. 351--368, Dec.
  2003.

\bibitem{Volodymyr2015Human}
M.~Volodymyr, K.~Koray, S.~David, A.~A. Rusu, V.~Joel, M.~G. Bellemare,
  G.~Alex, R.~Martin, A.~K. Fidjeland, and O.~Georg, ``Human-level control
  through deep reinforcement learning,'' \emph{Nature}, vol. 518, no. 7540, pp.
  529--533, Feb. 2015.

\bibitem{Li2019Deep}
H.~Li, K.~Ota, and M.~Dong, ``Deep reinforcement scheduling for mobile
  crowdsensing in fog computing,'' \emph{ACM Trans. Internet Technol.},
  vol.~19, no.~2, Apr. 2019.

\bibitem{zappone2019Model}
A.~{Zappone}, M.~{Di Renzo}, M.~{Debbah}, T.~T. {Lam}, and X.~{Qian},
  ``Model-aided wireless artificial intelligence: Embedding expert knowledge in
  deep neural networks for wireless system optimization,'' \emph{IEEE Veh.
  Technol. Mag.}, vol.~14, no.~3, pp. 60--69, Jul. 2019.

\bibitem{Zappone2019Wireless}
A.~{Zappone}, M.~{Di Renzo}, and M.~{Debbah}, ``Wireless networks design in the
  era of deep learning: Model-based, ai-based, or both?'' \emph{IEEE Trans.
  Commun.}, vol.~67, no.~10, pp. 7331--7376, Jun. 2019.

\bibitem{Rubinstein2008Simulation}
R.~Y. Rubinstein and D.~P. Kroese, \emph{Simulation and the Monte Carlo
  Method}.\hskip 1em plus 0.5em minus 0.4em\relax John Wiley \& Sons,, 2008.

\bibitem{bailyn2004generalized}
J.~F. Bailyn, ``Generalized inversion,'' \emph{Natural Language \& Linguistic
  Theory}, vol.~22, no.~1, pp. 1--50, 2004.

\bibitem{pineda1987generalization}
F.~J. Pineda, ``Generalization of back-propagation to recurrent neural
  networks,'' \emph{Physical Rev. Lett.}, vol.~59, no.~19, p. 2229, Jun. 1987.

\bibitem{Xu2015Block}
Y.~Xu and W.~Yin, ``Block stochastic gradient iteration for convex and
  nonconvex optimization,'' \emph{SIAM J. Optimization}, vol.~25, no.~3, pp.
  1686--1716, Jan. 2015.

\bibitem{Hirtenfelder2007Effective}
F.~{Hirtenfelder}, ``Effective antenna simulations using {CST MICROWAVE
  STUDIO$^\circledR$},'' in \emph{Proc. INICA}, Munich, Germany, Mar. 2007.

\bibitem{skiena2012sorting}
S.~S. Skiena, \emph{Sorting and searching}.\hskip 1em plus 0.5em minus
  0.4em\relax London, U.K.: Springer, 2012.

\bibitem{sipper1993serial}
M.~Sipper, ``A serial complexity measure of neural networks,'' in \emph{Proc.
  IEEE ICNN}, San Francisco, CA, Mar. 1993.

\bibitem{PETKOVIC2009270}
M.~D. Petković and P.~S. Stanimirović, ``Generalized matrix inversion is not
  harder than matrix multiplication,'' \emph{J. Comput. Appl. Math.}, vol. 230,
  no.~1, pp. 270 -- 282, Aug. 2009.

\end{thebibliography}
\end{footnotesize}

\end{document}